\documentclass[11pt]{article}
\def\submission{1} 

\usepackage[utf8]{inputenc}
\usepackage{setspace}
\usepackage{fullpage}
\usepackage{graphicx}
\usepackage{bbm}
\usepackage{colortbl}
\graphicspath{ {./figures/} }
\usepackage{subcaption}
\usepackage{amsmath,amsfonts,amssymb,amsthm}
\usepackage{array}
\usepackage{multirow}
\usepackage{xcolor}
\usepackage[colorlinks=true,
            citecolor=blue]{hyperref}
\usepackage{fullpage}
\usepackage{algorithm}
\usepackage{makecell}
\usepackage[noend]{algpseudocode}

\usepackage{framed}

\usepackage{bbm}
\usepackage{pgfplots}
\pgfplotsset{compat=1.16}
\usepackage{mathtools}
\usepackage{float}
\usepackage{indentfirst}
\usepackage{tikz}
\usepackage[capitalize]{cleveref}
\usepackage{accents}
\usepackage{hhline}
\usetikzlibrary{external, matrix, positioning, patterns, arrows.meta, decorations.markings, decorations.pathmorphing, calc, shapes.misc}
\tikzset{cross/.style={cross out, draw=black, minimum size=2*(#1-\pgflinewidth), inner sep=0pt, outer sep=0pt},
cross/.default={1pt}}
\allowdisplaybreaks
\newtheorem{theorem}{Theorem}[section]
\newtheorem{corollary}{Corollary}[section]
\newtheorem{lemma}{Lemma}[section]

\newtheorem{fact}{Fact}[section]
\newtheorem*{prop*}{Proposition}
\newtheorem*{lemma*}{Lemma}
\newtheorem{conjecture}{Conjecture}[section]
\newtheorem*{definition*}{Definition}
\newtheorem*{theorem*}{Theorem}

\theoremstyle{definition}
\newtheorem{definition}{Definition}[section]
\newtheorem{example}{Example}[section]
\theoremstyle{remark}
\newtheorem{remark}{Remark}[section]
\numberwithin{equation}{section}
\usepackage{xspace}
\crefname{prop}{Proposition}{Propositions}

\makeatletter
\AddToHook{cmd/appendix/before}{\def\cref@section@alias{appendix}\def\cref@subsection@alias{appendix}}
\makeatother

\newcommand{\supp}{\operatorname{supp}}
\newcommand{\rank}{\operatorname{rank}}

\newcommand{\tr}{\operatorname{Tr}}

\newcommand{\etal}{\emph{et al.\@}}

\newcommand{\accept}{\textsf{accept}}
\newcommand{\reject}{\textsf{reject}}

\newcommand{\Var}{\mathrm{Var}}

\newcommand{\vecmap}{\mathrm{vec}}

\newcommand{\Cov}{\mathrm{Cov}}
\newcommand{\cyc}{\mathrm{cyc}}

\newcommand{\QChi}[2]{\mathrm{D}_{\chi^2}\left(#1 \| #2\right)}
\newcommand{\QChiAlpha}[3]{\mathrm{D}_{\chi^2_{#1}}\left(#2 \| #3\right)}
\newcommand{\Inner}[2]{\langle #1, #2 \rangle}
\newcommand{\InnerNamed}[3]{\langle #1, #2 \rangle_{#3}}

\newcolumntype{C}{>{\centering\arraybackslash}p{1.3cm}}

\newcommand{\eqdef}{\coloneqq}

\newcommand{\eps}{\varepsilon}

\usepackage[backend=biber,style=alphabetic,maxbibnames=99,maxalphanames=99,url=false]{biblatex}
\providecommand{\email}[1]{\href{mailto:#1}{\nolinkurl{#1}\xspace}}

\title{Distributed Property Testing with (Quantum) Carrier Pigeons: Tight Bounds on State Certification}

\ifnum\submission=1 
\author{Kenny Chen\thanks{The University of Sydney.
Email: \email{kche5493@uni.sydney.edu.au}}}
\else 
\author{(Authors withheld)}
\fi
\date{June 2026}

\begin{document}
\maketitle 
\begin{abstract}
    Recently, Doosti \etal~introduced the problem of distributed quantum state verification, where $m$ distributed nodes are given a copy of an unknown state $\rho$, and can send limited one way communication to a central node, who has a complete description of a known state $\sigma$. They ask how many distributed nodes $m$ are required, before the central node can succeed at distinguishing whether $\rho=\sigma$ or $\|\rho-\sigma\|_1\geq\eps$ with high probability. In the setting where only quantum communication is allowed, Doosti \etal~exhibit conditional lower bounds in both the public and private-coin settings, and a matching upper bound in the public-coin setting. We extend these results, and show unconditional lower bounds for when both classical and quantum communication are permitted. We show the public-coin lower bound is tight by giving an algorithm with a matching upper bound. We also show an almost tight upper bound in the private-coin setting when only quantum communication is permitted.
\end{abstract}
\section{Introduction}
This paper is concerned with the task of (quantum) state certification: that is, the task of checking whether an unknown state $\rho$ is identical to a known state $\sigma$, or whether it is ``far'' from it. Formally, it is defined as follows:
\begin{definition}[Quantum State Certification]\label{def:quantum_state_certification}
    Given $\eps,\delta\in(0,1]$, a complete classical description of a quantum state $\sigma$, an algorithm is said to \emph{$(\eps,\delta)$-certify $\sigma$} if, given multiple copies of an unknown state $\rho$, it satisfies the following:
    \begin{enumerate}
        \item If $\rho=\sigma$, the algorithm outputs ``\accept'' with probability at least $1-\delta$;
        \item If $\|\rho-\sigma\|_1\geq\eps$, the algorithm outputs ``\reject'' with probability at least $1-\delta$.
    \end{enumerate}
    When $\delta$ is set to $1/3$, we simply say that the algorithm $\eps$-certifies $\sigma$. The number of copies of the unknown state $\rho$ (in the worst case over $\rho,\sigma$) is said to be the \emph{sample complexity} of the algorithm.
\end{definition}
The distance measure above is the trace distance, and while the computational complexity of the algorithm is important in its own right, the main focus is typically the sample complexity of the task as a function of $\eps,\delta$ and the dimension $d$ (defined as the minimum sample complexity of any algorithm for $(\eps,\delta)$-certification). This quantity is indeed of particular importance in the quantum setting, since, due to the no cloning theorem, obtaining a large number of a state $\rho$ is generally very difficult, so algorithms should aim to be very efficient in the number of copies required. This problem has been very well studied in a variety of settings, with tight bounds known, see \cite{DBLP:conf/stoc/ODonnellW15} and \cite{DBLP:conf/stoc/BadescuO019}, establishing that this task needs $\Theta\left(\frac{d}{\eps^2}\right)$ copies of $\rho$. However, these settings assume the tester themselves have access to the unknown state $\rho$. In practice, the verifier, who has $\sigma$, may be physically separated from the party, or parties, that have $\rho$. In this sense, it is not possible for the certifier themselves to perform operations on $\rho$, and they need to rely on other parties to do that processing for them. The results are then sent to the central node by the other parties, which then allows them to perform the certification. There is also nuance in this communication though, being that specifically quantum communication is difficult. Indeed, states are delicate and easily disturbed, and our quantum carrier pigeons may need to travel large distances, so we ought not burden them with more communications than is necessary. This raises the main question of this paper: 

\begin{framed}\itshape
\noindent Given several parties, each holding one copy of $\rho$ and located away from the main certifier, who has $\sigma$, with limited information being able to be transmitted to the central certifier, how many copies of $\rho$ are needed for them to perform quantum state certification successfully?
\end{framed}

Aiming to formalize this type of questions, Doosti \etal~\cite{doosti2026distributedquantumpropertytesting} recently introduced the setting of \emph{distributed quantum property testing}. Here copies of an unknown quantum state are given to multiple distributed nodes, all of which can process the quantum state, then communicate to a central referee node whose task is, after obtaining the information from the distributed nodes, to run a learning or testing algorithm to determine some property of the unknown quantum state. 
\begin{definition}[Distributed Quantum Inference]\label{def:distributed_quantum_inference}
    The $(n_c,n_q,R)$\footnote{Doosti \etal~actually define an even more general model where entanglement is allowed between distributed nodes, with an extra parameter $E$ quantifying the number of entangled qubits shared among distributed nodes. However, we never deal with entanglement in this paper, and hence omit it in this definition.} model of distributed quantum inference is the model where no communication is allowed between distributed nodes, one-way communication is allowed from distributed nodes to the central node, and:
    \begin{enumerate}
        \item Each of the distributed nodes may communicate at most $n_c$ classical bits and $n_q$ qubits to the central node.
        \item $R\in\{\text{public, private}\}$ denotes whether all nodes only have access to a private source of randomness (private-coin protocols), or if all nodes have access to a shared public source of randomness (public-coin protocols).
    \end{enumerate}
The resource we seek to minimize is the number of distributed nodes $m$ to solve the prespecified task at hand. It is worth discussing here the parameter regimes of $n_q$ and $n_c$ that make this problem interesting and nontrivial. We paraphrase this discussion from that of Doosti \etal~\cite{doosti2026distributedquantumpropertytesting}:
\begin{itemize}
    \item If $n_q\geq n$, then each distributed node can simply send the entire quantum state to the central node. Then, this is equivalent to the regular problem of quantum state certification, which is fully characterized by Badescu \etal~\cite{DBLP:conf/stoc/BadescuO019}.
    \item If $n_q=0$ but $n_c\geq n$, each distributed node can perform a single copy measurement, and send it to the classical node to postprocess. Then, this is equivalent to state certification with nonadaptive single-copy measurements, which is fully characterized by Chen \etal~\cite{DBLP:conf/focs/Chen0HL22}
    \item If $n_q=0$ and $n_c<n$, this setting is equivalent to the above single copy measurement setting, with additional communication restrictions. Namely, the measurement operators can only have $2^{n_c}$ outcomes. This setting is also fully characterized by Liu \etal~\cite{DBLP:journals/corr/abs-2408-17439}.
\end{itemize}
\end{definition}
Hence, the interesting parameter regime is the one where $n_c+n_q<n$, which is the one we focus on. More broadly, this problem is the quantum analogue of classical distributed property testing, and the quantum generalization of the framework of Acharya \etal~\cite{DBLP:journals/tit/AcharyaCLST22}.
This model, quite natural to consider in view of the increased interest in distributed quantum computing \cite{DBLP:journals/cn/CaleffiAFIMC24}, aligns with the motivation outlined above: and indeed, the flagship problem Doosti \etal~focused on to illustrate this distributed model is that of {quantum state certification}, as introduced above (and which itself is the quantum generalization of distributed \emph{identity testing} for classical probability distributions). 
Translating state certification into the distributed quantum inference setting, the $m$ distributed nodes each will be given a copy of the unknown state $\rho$, whilst the central referee node will be given a complete classical description of the state $\sigma$. Then, the question is what the bounds on $m$ (number of nodes, and, equivalently, sample complexity) are for protocols for this task\footnote{In what follows, and consistent with the literature on distributed inference, we will interchangeably use the words \emph{algorithms} and \emph{protocols}.} such that the central node can succeed at distinguishing the two cases.\medskip

In this paper, we make significant progress towards resolving this question, essentially settling it for all but one setting. Before stating our results, we summarize the prior work of Doosti \etal\ First, in the public-coin setting, where no classical communication is allowed, they establish the following upper bound:\footnote{We hereafter focus on $\eps$-certification and state all results for fixed $\delta=1/3$, as upper bounds for $(\eps,\delta)$-certification then immediately follow by a standard majority vote, at the cost of a $O(\log(1/\delta))$ factor in the sample complexity.}
\begin{theorem}[{\cite[Theorem 1.3]{doosti2026distributedquantumpropertytesting}}]\label{thm:doosti_public_upper_bound}
    Fix any $1\leq n_q\leq \log d$. There is a protocol for $\eps$-certification of $d$-dimensional qudit states in the $(0,n_q,\text{public})$ model with sample complexity $O\left(\frac{d^2}{2^{n_q}\eps^2}\right)$.
\end{theorem}
Notice this theorem is stated for \emph{qudit} states, so the restriction on $n_q$ is equivalent to the discussion above for qubit states. Turning to lower bounds, Doosti \etal~introduce the notion of \emph{mixedness preservation}: roughly speaking, a distributed node is mixedness-preserving if, when given as input the maximally mixed state, the message sent to the central node is also a maximally mixed state (potentially of a different dimension). Under the restriction that nodes are mixedness-preserving, they obtain the following two lower bounds:
\begin{theorem}[{\cite[Theorem 1.4]{doosti2026distributedquantumpropertytesting}}]\label{thm:doosti_public_lower_bound}
    Any protocol for $\eps$-certification of $d$-dimensional states in the $(0,n_q,\text{public})$ model \emph{with mixedness-preserving nodes} must have sample complexity $\Omega\left(\frac{d^2}{2^{n_q}\eps^2}\right)$.
\end{theorem}
\begin{theorem}[{\cite[Theorem 1.5]{doosti2026distributedquantumpropertytesting}}]\label{thm:doosti_private_lower_bound}
    Any protocol for $\eps$-certification of $d$-dimensional states in the $(0,n_q,\text{private})$ model \emph{with mixedness-preserving nodes} must have sample complexity $\Omega\left(\frac{d^3}{2^{2n_q}\epsilon^2}\right)$.
\end{theorem}
The authors then conjecture that the assumption of mixedness preservation can be removed:
\begin{conjecture}[\cite{doosti2026distributedquantumpropertytesting}]\label{conj:doosti}
    The lower bounds of \cref{thm:doosti_public_lower_bound} and \cref{thm:doosti_private_lower_bound} are tight for \emph{all} protocols, not only mixedness-preserving ones.
\end{conjecture}
\subsection{Our Results}
We improve on the above results in 3 ways: (1) by considerably generalizing the above lower bounds, (2) by obtaining a (near)-matching private-coin upper bound in the quantum communication setting, and (3) by providing  an optimal public-coin upper bound in the mixed (classical and quantum) communication one. 

Indeed, our first results answer \cref{conj:doosti} in the affirmative, removing the requirement of mixedness preservation for both public- and private-coin protocols:
\begin{theorem}\label{thm:public_coin_no_classical_lb}
    Any protocol for $\eps$-certification of $d$-dimensional states in the $(0,n_q,\text{public})$ model must have sample complexity $\Omega\left(\frac{d^2}{2^{n_q}\eps^2}\right)$.
\end{theorem}
\begin{theorem}\label{thm:private_coin_no_classical_lb}
    Any protocol for $\eps$-certification of $d$-dimensional states in the $(0,n_q,\text{private})$ model must have sample complexity $\Omega\left(\frac{d^3}{2^{2n_q}\eps^2}\right)$.
\end{theorem}
In fact, these two results are a special case of stronger lower bounds we establish, in the setting where both classical \emph{and} quantum communication are allowed. This setting, while defined in~\cite{doosti2026distributedquantumpropertytesting}, was not addressed in their work.
\begin{theorem}[General Public-Coin Lower Bound]\label{thm:public_coin_lb}
    Any protocol for $\eps$-certification of $d$-dimensional states in the $(n_c,n_q,\text{public})$ model must have sample complexity $\Omega\left(\frac{d^2}{2^{ n_q+n_c/2}\eps^2}\right)$.
\end{theorem}
\begin{theorem}[General Private-Coin Lower Bound]\label{thm:private_coin_lb}
    Any protocol for $\eps$-certification of $d$-dimensional states in the $(n_c,n_q,\text{private})$ model must have sample complexity $\Omega\left(\frac{d^3}{2^{2n_q+n_c}\eps^2}\right)$.
\end{theorem}
 We then complement these negative results with (nearly) matching upper bounds. First, we go back to the private-coin setting with no classical communication. for which we provide a protocol with sample complexity tight up to a single logarithmic factor:
\begin{theorem}[Private-Coin Upper Bound]\label{thm:private_coin_no_classical_ub}
    Fix any $1\leq n_q\leq \log d$. There is a protocol for $\eps$-certification of $d$-dimensional states in the $(0,n_q,\text{private})$ model with sample complexity $O\left(\frac{d^3}{2^{2n_q}\eps^2}\cdot \log d\right)$.
\end{theorem}
Next, we turn to the public-coin setting, with \emph{both} classical and quantum communication permitted, for which we give a sample-optimal algorithm:
\begin{theorem}[General Public-Coin Upper Bound]\label{thm:public_coin_ub_inf}
    Fix any $1\leq n_q+n_c\leq \log d$. There is a protocol for $\eps$-certification of $d$-dimensional states in the $(n_c,n_q,\text{private})$ model with sample complexity $O\left(\frac{d^2}{2^{n_q+n_c/2}\eps^2}\right)$.
\end{theorem}
We summarize these results in \cref{tab:results_table}:

\begin{table}[H]
    \centering
    \scriptsize
    \resizebox{\textwidth}{!}{%
    \begin{tabular}{|p{2cm}|l|c|c|c|c|}
    \hhline{|-|-|-|-|-|-|}
    \multicolumn{2}{|l|}{} & \multicolumn{2}{c|}{Public-Coin} & \multicolumn{2}{c|}{Private-Coin} \\ \hhline{|-|-|-|-|-|-|}
    \multirow{2}{=}{Quantum Communication Only}
     & Upper Bound & \multicolumn{2}{c|}{\cellcolor{gray!20}$O\left(\frac{d^2}{2^{n_q}\eps^2}\right)$} & \multicolumn{2}{c|}{$O\left(\frac{d^3}{4^{n_q}\eps^2}\log d\right)$} \\ \hhline{|~|-|-|-|-|-|}
     & Lower Bound & \cellcolor{gray!20}$\Omega\left(\frac{d^2}{2^{n_q}\eps^2}\right)^\dagger$ & $\Omega\left(\frac{d^2}{2^{n_q}\eps^2}\right)^*$ & \cellcolor{gray!20}$\Omega\left(\frac{d^3}{4^{n_q}\eps^2}\right)^\dagger$ & $\Omega\left(\frac{d^3}{4^{n_q}\eps^2}\right)^*$ \\ \hhline{|-|-|-|-|-|-|}
    \multirow{2}{=}{Quantum and Classical Communication}
     & Upper Bound & \multicolumn{2}{c|}{$O\left(\frac{d^2}{2^{n_q+n_c/2}\eps^2}\right)$} & \multicolumn{2}{c|}{} \\ \hhline{|~|-|-|-|-|-|}
     & Lower Bound & \multicolumn{2}{c|}{$\Omega\left(\frac{d^2}{2^{n_q+n_c/2}\eps^2}\right)$} & \multicolumn{2}{c|}{$\Omega\left(\frac{d^3}{2^{2n_q+n_c}\epsilon^2}\right)$} \\ \hhline{|-|-|-|-|-|-|}
    \end{tabular}%
    }
    \caption{Upper and lower bounds for both public- and private-coin models. The shaded cells represent previously known results by Doosti \etal, and the unshaded cells represent our results. The $\dagger$ indicates a conditional bound, and the $*$ represents that the result is a special case of a more general result we establish.}
    \label{tab:results_table}
\end{table}
\subsection{Our Techniques}
\subsubsection{Lower Bound Techniques}
We start by giving an overview of the methodology of Doosti \etal, which builds upon and extends the classical counterpart developed by~\cite{DBLP:journals/tit/AcharyaCT20}, and which we use as inspiration for our techniques. The starting point is that $n_q$ qubits can be instead represented as one $d_q$-dimensional qudit, where $d_q \eqdef 2^{n_q}$. Hence, without loss of generality, we assume from now on that the distributed nodes are communicating $d_q$-dimensional qudits. At a high level, Doosti \etal~show it is hard to distinguish the maximally mixed state from a state close to the maximally mixed state, by explicitly constructing such a hard family. Then they analyze the $\chi^2$-divergence\footnote{See \cref{def:quantum_chi_squared} for a definition.} between the qudits received, which upper bounds their trace distance but is more tractable to analyze. After handling a fair amount of technicalities, they are able to reduce this to the problem of lower bounding the norm of a certain operator they define. This operator depends on a quantum channel $\Phi$ (which can be thought of the quantum analogue of a classical communication channel) which the distributed nodes use to send qudits to the central referee node. Crucially, one key step of their argument requires to evaluate $\Phi(I/d)^{-1}$, that is, the inverse of the state which the quantum channel maps the maximally mixed state to. This is where the mixedness preservation assumption comes in: indeed, under this assumption on $\Phi$, this quantity is simply the inverse of another maximally mixed state, which becomes straightforward to evaluate. However, it was unclear how to handle this inverse for general quantum channels.

Our approach to remove the mixedness preserving assumption relies on a few key observations. First, we use a different, \emph{symmetrized} version of $\chi^2$-divergence, which has some more amenable properties for the proof, and also may be of independent interest. We then must re-establish several key results for this symmetrized $\chi^2$ divergence, in order for the overall argument to go through. After this, this symmetrized divergence allows us to proceed as follows: denoting, for an arbitrary quantum channel\footnote{See \cref{sec:quantum_channels_instruments} for a definition.} $\Phi(X)$, the output of the maximally mixed state under the channel $\Phi$ as $\tau \eqdef \Phi(I/d)$, we define an ``augmented'' channel $\Psi(X) \eqdef \tau^{-1/4}\Phi(X)\tau^{-1/4}$. One can show $\Psi(X)$ is well-defined, and analyzing this channel instead paves the way towards establishing \cref{thm:public_coin_no_classical_lb} and \cref{thm:private_coin_no_classical_lb}.

While doing so would enable us to obtain these two theorems, we do not actually go through with this augmented channel $\Psi$ directly, as our goal is to instead establish the more general versions allowing classical communication as well. To do so, the next insight is that instead of using a quantum \emph{channel}, as Doosti \etal~do, one can instead use a \emph{quantum instrument},\footnote{See \cref{sec:quantum_channels_instruments} for the definition.} which allows for communication of both quantum qudits and classical bits. With the notion of quantum instrument in hand, the idea is then similar: one can define an augmented quantum instrument in the spirit of the above paragraph, prove it is well-defined with nice properties, and then follow the framework of Doosti \etal~to reprove lower bounds using this augmented quantum instrument, instead of a quantum channel. Doing so yields \cref{thm:public_coin_lb} and \cref{thm:private_coin_lb}. We emphasize that although Doosti \etal~do provide the blueprint for the proof, much of the work is adapting their framework to fit these augmented quantum instruments, which is a nontrivial task. The results we establish for quantum instruments are novel, and also may be of independent interest for other problems. To the best of our knowledge this is the first application of quantum instruments in the learning setting.
\subsubsection{Upper Bound Techniques}
The private-coin upper bounding framework starts from a similar point as Doosti \etal\ However, instead of using the Hilbert-Schmidt certifier, which, given $m$ copies of two states, tests whether two states are close or far in terms of Hilbert--Schmidt distance, we instead use the Hilbert-Schmidt estimator, a subroutine of Badescu \etal~for the certifier, which, given $m$ copies of the two states, directly outputs an estimate of the trace distance between two states. Then the high-level idea is the same, where we aim to estimate the Hilbert-Schmidt distance after each state $\rho$ has been rotated by a unitary, and hope that the Hilbert-Schmidt distance for when $\rho=\sigma$ and when $\rho$ is far from $\sigma$ are separated. The key difficulty is that in the public-coin case, every node is aware of the other unitaries used by the other nodes, but this is not true of the private-coin case. Instead, we show the existence of a set of ``good'' unitaries, such that it preserves, up to constant factors, the Hilbert-Schmidt distance between $\rho$ and $\sigma$ (see \cref{thm:good_unitaries}). Then, the nodes in the private-coin setting can pre-agree on these unitaries, leading to a nearly tight \emph{deterministic} protocol for the private-coin setting.

The general public-coin algorithm goes in four main steps, which are explained further in \cref{sec:outline}. First, we do a more specific analysis of the Badescu \etal~'s Hilbert--Schmidt estimator, to show that when quantum instruments spread out over all the available classical strings, this corresponds to a lower variance. To be specific, whilst Badescu \etal's variance bound depends inversely on $m$, we derive a bound that depends inversely on $ms$, where $m$ is the number of copies and $s$ is the number of classical strings that can be output by the quantum instrument. This is done mainly by following the structure of their proof, but adapted for the specific case of classical-quantum states. Even in this specific setting, we do not claim a generic improvement, but instead identify sufficient conditions where we can get such an improvement. Specifically, if $s$ is constant, meaning the quantum instruments use few classical responses, then the improvement to variance will get absorbed into the leading constant factors; however, if $s$ is large enough, this will lead to a genuine improvement to the variance. In particular, if $s\approx K$, then we will be able to get variance that depends inversely on $K$. With this in mind, we design quantum instruments based on sampling Haar random unitaries such that they do spread their responses over all the possible classical strings, corresponding to the setting where $s\approx K$, which we call \emph{Haar instruments}. In other words, they can possibly respond with many classical strings of length $n_c$, with roughly equal probability, rather than responding with a few with high probability. These may also be of independent interest. This then allows us to realize the improvements to the Hilbert-Schmidt distance estimator, yielding a lower variance. We also then show this quantum instrument does not contract the Hilbert-Schmidt distance too much, making it still possible to distinguish $\rho=\sigma$ from $\|\rho-\sigma\|_1\geq\eps$. The actual algorithm then, samples Haar random unitaries to design these quantum instruments. Each distributed node then sends a quantum state to the central node through the quantum instruments. The central node, with full knowledge of $\sigma$, as well as the quantum instruments used, then estimates the distance between $\sigma$ passed through the quantum instruments, and the received states. If $\rho=\sigma$, most of these estimates will be below some threshold, and the central node can accept. Else, the Hilbert--Schmidt distances will be large, and hence most of the estimates large (with high probability), leading to the central node rejecting.

\subsection{Related Work}
As mentioned above, the setting of this work is a generalization of classical distributed property testing, such as studied in \cite{DBLP:journals/tit/AcharyaCT20} and \cite{DBLP:journals/tit/AcharyaCLST22}, in which they obtain tight bounds in the classical setting. For more exposition on how these particular bounds relate to this work, as well as discussion on quantum (nondistributed) state testing, we refer to the introduction by Doosti \etal~\cite{doosti2026distributedquantumpropertytesting}.

There has also been much growing interest in distributed quantum algorithms. In this setting, the problem occurs on a network, usually represented by a graph, where every node represents an instance of an algorithm. Every round, each node can perform computations, receive messages from neighboring nodes, or send messages to them, with the goal to compute something about the graph. In this sense, one can think of distributed property testing as a problem occurring on the star graph with one round of interaction allowed, which permits the distributed nodes to send information to the central node. It particularly mirrors the CONGEST model, where communication between the nodes in a round are limited (see \cite{355459} for more on the CONGEST model). Many problems have been studied in this distributed model, such as leader election, MST, and cycle detection, among others \cite{DBLP:journals/corr/abs-2602-15529}, \cite{10.1145/3662158.3662767}.

\section{Preliminaries}
Throughout, we use standard asymptotic notation, as well as $\tilde{O}(\cdot)$ to hide polylogarithmic factors in the argument. We start with some key definitions, used throughout the paper.
\begin{definition}[Schatten $p$ norms]
    For matrices/linear operators $T$ and $p\in[1,\infty]$, the \emph{Schatten $p$ norm} of $T$ is given by
    \begin{equation*}
        \|T\|_p=[\tr|T|^p]^{1/p},
    \end{equation*}
    where $|T|=\sqrt{T^\dagger T}$, for $p<\infty$. For $p=\infty$, this is defined as $\|T\|_\infty=\max_{\|u\|=1}\|Tu\|$, where the vector norm is the standard Euclidean norm.
\end{definition}
\begin{definition}[Quantum $\chi^2$ Divergence, \cite{TemmeDivergence}]\label{def:quantum_chi_squared}
    Let $\rho,\sigma\in\mathbb{C}^{d\times d}$. When $\supp(\rho)\subseteq\supp(\sigma)$, the \emph{$\alpha$-$\chi^2$ divergence between $\rho$ and $\sigma$} is defined as 
    \begin{equation*}
        \QChiAlpha{\alpha}{\rho}{\sigma}=\tr[(\rho-\sigma)\sigma^{-\alpha}(\rho-\sigma)\sigma^{\alpha-1}]=\tr[\rho\sigma^{-\alpha}\rho\sigma^{\alpha-1}]-1.
    \end{equation*}
\end{definition}
\noindent In particular, \cite{doosti2026distributedquantumpropertytesting} use this divergence when $\alpha=0$, i.e., 
    \begin{equation*}
        \QChiAlpha{0}{\rho}{\sigma}=\tr[(\rho-\sigma)^2\sigma^{-1}].
    \end{equation*}
In contrast, we will use the ``symmetrized'' version of this $\chi^2$ divergence, which is when $\alpha=\frac{1}{2}$. From this point on, unless specifically pointed out otherwise, we will use this symmetrized version and write $\QChi{\rho}{\sigma}$ for $\QChiAlpha{1/2}{\rho}{\sigma}$. We then have
\begin{align}
    \QChi{\rho}{\sigma}&=\tr[(\rho-\sigma)\sigma^{-1/2}(\rho-\sigma)\sigma^{-1/2}]\notag\\
    &=\tr[\sigma^{-1/2}(\rho-\sigma)\sigma^{-1/2}(\rho-\sigma)]\notag\\
    &=\tr[(\sigma^{-1/4}(\rho-\sigma)\sigma^{-1/4})^2] \label{eq:useful:symmetrization}\\
    &=\tr[\rho\sigma^{-1/2}\rho\sigma^{-1/2}]-1. \notag
\end{align}
Notice the third equality, \eqref{eq:useful:symmetrization}, is what makes this ``symmetric'', since one can collapse the trace into a square. We have that this symmetrized version still upper bounds the trace distance:
\begin{lemma}[{\cite[Lemma 5]{TemmeDivergence}}]\label{lemma:1_norm_qchi_inequality}
    Let $\rho$ and $\sigma$ be two quantum states. Then, $\|\rho-\sigma\|_1\leq\sqrt{\QChi{\rho}{\sigma}}$.
\end{lemma}

We also briefly recall what a \emph{classical quantum  (CQ)  state} is. Hereafter, we denote the ``classical register'' by $C$ and the ``quantum register'' by $Q$, with dimension $K$ and $d_q$ respectively. Then, a classical quantum state is a special quantum state of the following, specific form:
\begin{equation*}
    \omega_{CQ}=\sum_{c=1}^K\omega_c\otimes|c\rangle\langle c|_C,
\end{equation*}
where each $\omega_c$ is a $d_q$-dimensional positive semidefinite operator with trace at most $1$. That is, it alone is not a valid quantum state, but the overall state $\omega_{CQ}$ has trace 1, and hence is a valid state.

Finally, before turning to the definition of quantum channels and instruments, we recall a useful probabilistic inequality, due to Paley and Zygmund:
\begin{lemma}[Paley-Zygmund inequality; see, e.g.,\cite{steele2004paley}]\label{lemma:Paley_Zygmund}
    Let $Y\geq 0$ be a random variable with finite first moment. Then, for any $\theta\in [0,1]$, we have
    \begin{equation*}
        \Pr[Y\geq \theta\mathbb{E}[Y]]\geq (1-\theta)^2\frac{(\mathbb{E}[Y])^2}{\mathbb{E}[Y^2]}.
    \end{equation*}
\end{lemma}

\subsection{Quantum Channels and Quantum Instruments}\label{sec:quantum_channels_instruments}
In this section, we establish the main tools used by the distributed nodes to communicate with the central referee node. Firstly, when only quantum communication is allowed, each node can design a \emph{quantum channel} $\Phi$, which they can use to send a $d_q$-dimensional qudit to the central referee node.
\begin{definition}[Quantum Channel]\label{def:quantum_channel}
    A map $\Phi\colon\mathbb{C}^{d\times d}\to\mathbb{C}^{d'\times d'}$ is a \emph{quantum channel} if it is a completely positive trace preserving map (CPTP map).
\end{definition}
Analogously to how bits can be transmitted across classical channels, these are the quantum generalization that allows transmissions of quantum states. By design, quantum channels cannot communicate classical bits: the generalization that allows this is a \emph{quantum instrument}.
\begin{definition}[Quantum Instrument, \cite{hashim2026understandingquantuminstruments}]\label{def:quantum_instrument}
    A \emph{quantum instrument} (QI) $\mathcal{I}\colon\mathbb{C}^{d\times d}\to \mathbb{C}^{d_q\times d_q}\otimes \mathbb{C}^{K\times K}$ is a completely positive trace preserving map:
    \begin{equation*}
        \mathcal{I}(\rho)=\sum_{i\in K}\mathcal{E}_i(\rho)\otimes|i\rangle\langle i|,
    \end{equation*}
    where each $\mathcal{E}_i$ is a completely positive process such that the probability a measurement outcome $i$ is observed given a state $\rho$ is given by
    $
    \Pr[i\mid \rho]=\tr[\mathcal{E}_i(\rho)].
    $   
    Then, the post-measurement state, conditioned on the classical outcome $i$, is 
    \begin{equation*}
        \rho_i\eqdef\frac{\mathcal{E}_i(\rho)}{\tr[\mathcal{E}_i(\rho)]}.
    \end{equation*}
\end{definition}
We make a few comments on the above definition before continuing. First, notice that although each individual $\mathcal{E}_i$ is not necessarily a quantum channel (since it may not preserve traces) we require $\mathcal{I}$ overall to preserve traces. Then, we have that 
\begin{equation*}
    \tr[\rho]=\sum_i\tr[\mathcal{E}_i(\rho)]=1.
\end{equation*}
Furthermore, the post-measurement state $\rho_i$, conditioned on seeing the classical outcome $i$, is a valid quantum state. Second, when we use ``measurement'' in the above definition, it is better interpreted in the colloquial sense of ``extracting information from a quantum system,'' and not the mathematical formalism of applying a measurement on the qubit system (though one can think of it as the latter). We give some examples to help illustrate quantum instruments, and hopefully make the definition and comments clear.
\begin{example}
    As a sanity check, let us assume the quantum instrument has $K=1$ (corresponding to $n_c=0$), so no classical communication is allowed. Then, we have that $\mathcal{I}(\rho)=\mathcal{E}(\rho)$. In this case, we have that $\mathcal{E}$ is in fact a CPTP map, and this reduces to a quantum channel, as is expected.
\end{example}
\begin{example}
    To illustrate the second point, let us assume a distributed node receives the state $\rho$, and does the following. For the classical string, they sample an $n_c$ bit string $s$ uniformly at random, and send that. Then, they use some quantum channel $\Phi$, independent of the random string $s$, and send the quantum qudit $\Phi(\rho)$, meaning they send the pair $\{s, \Phi(\rho)\}$. Then, denoting $K=2^{n_c}$, each string $s$ appears with probability $1/K$. Correspondingly, for each of the strings, they have $\mathcal{E}_s(\rho)=\frac{1}{K}\Phi(\rho)$. Hence, the full quantum instrument is
    \begin{equation*}
        \mathcal{I}(\rho)=\sum_{i\in\{0,1\}^{n_c}}\mathcal{E}_i(\rho)\otimes|i\rangle\langle i|.
    \end{equation*}
    Then the referee node, who receives the output $\mathcal{I}(\rho)$, can measure the classical register, and see one string $i=s$ with probability $1/K$. Then, conditioned on seeing that string $s$, what the referee node has remaining is the quantum state
    \begin{equation*}
        \frac{\frac{1}{K}\Phi(\rho)}{\tr[\frac{1}{K}\Phi(\rho)]}=\Phi(\rho).
    \end{equation*}
    Hence, the referee node sees a classical string uniformly at random, and the output of $\rho$ fed into the quantum channel, which is exactly the process we described above.
\end{example}
\begin{example}
    Now, instead of seeing a random quantum string, say regardless of whatever $\rho$ is, the distributed nodes will send the same classical string $s$, and use the quantum channel $\Phi$ for quantum communication. Then, we have that $\mathcal{E}_s(\rho)=\Phi(\rho)$, and $\mathcal{E}_i(\rho)=0$, for all $i\neq s$, with some abuse of notation for the dimension of the zero state in the latter. Then, we have that the corresponding quantum instrument is
    \begin{equation*}
        \mathcal{I}(\rho)=\Phi(\rho)\otimes|s\rangle\langle s|.
    \end{equation*}
    Then, the referee receives the output of the quantum instrument and can do the same: they measure the classical register to see the classical string $s$, and what remains is the output of the channel $\Phi(\rho)$.
\end{example}
It will also be helpful to outline some alternative ways of expressing these channels and instruments. One of these is the \emph{Liouville representation}.
\begin{definition}[Vectorization, Liouville Matrix Representation]\label{def:vectorization}
    The \emph{vectorization} map $\vecmap :\mathbb{C}^{d\times d}\to\mathbb{C}^{d^2}$ flattens a matrix into a vector. This naturally preserves inner product, where if $X$ and $Y$ are matrices, then $\InnerNamed{X}{Y}{HS}=\Inner{\vecmap(X)}{\vecmap(Y)}$, where the first inner product is the Hilbert--Schmidt inner product ($\tr[X^\dagger Y]$) and the latter is the inner product on vectors.

    The \emph{Liouville matrix} of a quantum channel $\Phi:\mathbb{C}^{d\times d}\to\mathbb{C}^{d'\times d'}$ is the unique matrix $M_\Phi\in\mathbb{C}^{d'^2\times d^2}$ such that $\vecmap(\Phi(X))=M_\Phi\vecmap(X)$.
\end{definition}
The other is the \emph{Kraus decomposition}.
\begin{definition}[Kraus Operator Decomposition]
    Any completely positive operator $\Phi:\mathbb{C}^{d\times d}\to\mathbb{C}^{d'\times d'}$ can be described by a set of operators $\{A_k\}\subseteq\mathbb{C}^{d'\times d}$, where $k<dd'$, such that the action of the channel is given by $\Phi(X)=\sum_k A_kXA_k^\dagger$, and $\sum_k A_k^\dagger A_k=I$.
\end{definition}
The existence of the above decomposition is known as the Choi-Kraus theorem, a standard tool: see for instance \cite[Chapter 4]{Wilde_2016}.
\section{Lower Bound Approach}
In this section we lay out some key results and describe our lower bound approach. First, as pointed out previously, we will focus on proving \cref{thm:public_coin_lb} and \cref{thm:private_coin_lb}, the lower bounds in the setting where both classical and quantum communication are allowed. By simply setting the classical communication $n_c=0$, we recover \cref{thm:public_coin_no_classical_lb} and \cref{thm:private_coin_no_classical_lb}. 

We start by describing the hard instances we will use in our lower bounds. These are the same ones as Doosti \etal~and previous work, and consist of ``small perturbations around the maximally mixed state'':
\begin{definition}\label{def:hard_instances}
    Let $\{V_i\}_{i\in[d^2]}$ form an orthonormal basis for $\mathbb{C}^{d\times d}$ w.r.t. the Hilbert-Schmidt (HS) inner product, with $V_{d^2}\eqdef I_d/\sqrt{d}$. Then for some integer $d^2/2\leq \ell\leq d^2-1$, and $z\in\{-1,1\}^\ell$, we define the perturbation as
    \begin{equation*}
        \Delta_z=\frac{c\eps}{\sqrt{d}}\cdot\frac{1}{\sqrt\ell}\sum_{i=1}^\ell z_iV_i,\;N_z\eqdef\min\left\{1,\frac{1}{d\|\Delta_z\|_\infty}\right\},\;\bar{\Delta}_z\eqdef\Delta_z\cdot N_z\,,
    \end{equation*}
    where $c>0$ is an absolute constant. 
    In other words, $\Delta_z$ is a perturbation in each of the directions of the $\ell$ orthonormal basis matrices in the direction given by $z$, and $\bar\Delta_z$ caps the perturbation so that it does not blow the maximally mixed state too much. Then, we define the state $\rho_z=I/\sqrt d+\bar\Delta_z$. For convenience, we also denote $\mathcal{V}=[\vecmap(V_1),\dots,\vecmap(V_\ell)]\in\mathbb{C}^{d^2\times \ell}$.
\end{definition}
It was shown by Liu \etal~\cite{pmlr-v247-liu24a} that for any choice of bases $\{V_i\}$ and a uniformly random $z\sim\{-1,1\}^\ell$, the state $\rho_z$ is always a valid quantum state, and is in fact $\eps$-far from the maximally mixed state with high probability. We will describe such a distribution over states as \emph{almost-$\eps$ perturbations}.
\begin{definition}[Almost-$\eps$ Perturbations]\label{def:almost_epsilon_perturbations}
    An ensemble $D$ of quantum states is an almost-$\eps$ perturbation of a state $\sigma$ if $\Pr_{\rho\sim D}[\|\rho-\sigma\|_1\geq\eps]\geq \frac{1}{2}$. Let the set of all almost-$\eps$ perturbations of $\sigma$ be $\mathcal{D}_\eps(\sigma)$.
\end{definition}
The main tool we will use for lower bound is the following lemma:
\begin{lemma}\label{lemma:instrument_minimax_maximin_bounds}
    Let $n_c\geq 0$, $d\geq d_q\geq 2$. Then, for any public-coin protocol to succeed at $\eps$-certifying $d$-dimensional states using only $n_c$ length classical messages and $d_q$-dimensional quantum messages, the number of distributed nodes $m$ must be large enough such that 
    \begin{equation*}
        \min_{D\in\mathcal{D}_\eps(I_d/d)}\max_{\mathcal{I}_1,...,\mathcal{I}_m}\QChi{\mathbb{E}_{\rho\sim D}\left[\bigotimes_{i=1}^m\mathcal{I}_i(\rho)\right]}{\bigotimes_{i=1}^m\mathcal{I}_i(I_d/d)}\geq\frac{1}{16},
    \end{equation*}
    where each $\mathcal{I}_i$ is a quantum instrument outputting a length $n_c$ classical string and a $d_q$-dimensional quantum state. On the other hand, for any private-coin protocol to succeed in the same setting, we must have:
    \begin{equation*}
        \max_{\mathcal{I}_1,...,\mathcal{I}_m}\min_{D\in\mathcal{D}_\eps(I_d/d)}\QChi{\mathbb{E}_{\rho\sim D}\left[\bigotimes_{i=1}^m\mathcal{I}_i(\rho)\right]}{\bigotimes_{i=1}^m\mathcal{I}_i(I_d/d)}\geq\frac{1}{16}.
    \end{equation*}
\end{lemma}
This is the quantum instrument equivalent of \cite[Lemma 4.2]{doosti2026distributedquantumpropertytesting}. The key point is that \cref{lemma:instrument_minimax_maximin_bounds} allows us to turn the task of lower bounding the complexity into an equivalent one of upper bounding the $\chi^2$ distance. The approach is adapting that of Doosti \etal, which itself is a result of generalizing the approach of Acharya \etal~\cite{DBLP:journals/tit/AcharyaCLST22} and  Liu \etal~\cite{DBLP:journals/corr/abs-2401-09650}. We include a proof in \cref{app:lemma_instrument_minimax_maximin_bounds}. The next tool we need is the equivalent of the quantum Ingster-Suslina lemma for our symmetrized $\chi^2$-divergence.
\begin{lemma}[Symmetrized Quantum Ingster Suslina]\label{lemma:quantum_ingster-suslina_symm}
    Let $\theta$ be a random variable that parametrizes states $\rho_{i,\theta}\in \mathbb{C}^{d\times d}$ for all $i\in[m]$. Let $\sigma_i\in\mathbb{C}^{d\times d}$ be quantum states. Let
    \begin{equation*}
        \rho_\theta^{(m)}=\bigotimes_{i=1}^m\rho_{i,\sigma},~\sigma^{(m)}=\bigotimes_{i=1}^m\sigma_i.
    \end{equation*}
    Then,
    \begin{equation*}
        1+\QChi{\mathbb{E}_\theta[\rho_\theta^{(m)}]}{\sigma^{(m)}}=\mathbb{E}_{\theta,\theta'}\left[\prod_{i=1}^m(1+Z_i(\theta,\theta'))\right]\leq\mathbb{E}_{\theta,\theta'}\left[\exp\left(\sum_{i=1}^m Z_i(\theta,\theta')\right)\right],
    \end{equation*}
    where $\theta,\theta'$ are i.i.d., and
    \begin{equation*}
        Z_i(\theta,\theta')\eqdef \tr[\sigma_i^{-1/2}(\rho_{i,\sigma}-\sigma_i)\sigma_i^{-1/2}(\rho_{i,\sigma'}-\sigma_i)].
    \end{equation*}
\end{lemma}
This proof requires a different $Z_i$ but follows the same arguments as O'Donnell \etal (for the $0$-$\chi^2$ divergence), and we include it for completeness in \cref{app:symmetrized_quantum_ingster_suslina}.\footnote{In fact, we believe this general proof method will work for all $\alpha$-$\chi^2$ divergences, though one will get a different representation of $Z_i$ for each.} Note although the proof and statement look very similar, it is exactly this different $Z_i$ term that allows us to remove the mixedness preserving assumption, and prove our more general results. The high level plan is now clear: we will use \cref{lemma:instrument_minimax_maximin_bounds} to bound the $\chi^2$ divergence, and thus bound the complexity. In bounding the $\chi^2$ divergence, we note that by \cref{lemma:quantum_ingster-suslina_symm}, this is now equivalent to finding a bound on the $Z_i$ terms. \medskip

In summary, \cref{def:hard_instances} defines a class of almost-$\eps$ perturbations of the maximally mixed state, while \cref{lemma:instrument_minimax_maximin_bounds} states that the central referee node cannot distinguish the maximally mixed state from almost-$\eps$ perturbations of the maximally mixed state unless the number of distributed nodes $m$ is large enough, and~\cref{lemma:quantum_ingster-suslina_symm} provides a way to bound the terms in~\cref{lemma:instrument_minimax_maximin_bounds} to quantify what ``large enough'' means.\medskip

Before moving on to the proof of the lower bound itself, we require one last lemma about quantum instruments.
\begin{lemma}\label{lemma:instrument_support}
    Let $\mathcal{I}\colon\mathbb{C}^{d\times d}\to\mathbb{C}^{d_q\times d_q}\otimes \mathbb{C}^{K\times K}$ be a quantum instrument, and $\tau=\mathcal{I}(I/d)$. Let $P=\Pi_{\supp(\tau)}$. Then
    \begin{equation*}
        \mathcal{I}(X)=P\mathcal{I}(X)P.
    \end{equation*}
\end{lemma}
\begin{proof}
     Let $F=\mathcal{I}(I)=d\tau$, and $P=\Pi_{\supp(F)}$.We will show that $\mathcal{I}(X)$ has no support on $\ker F$, meaning it is completely defined on the support of $F$. To do so, we write a Kraus operator decomposition of the instrument $\mathcal{I}$ given by $\{A_k\}$, which exists since a quantum instrument is a completely positive operator. Then
     \begin{equation*}
         \mathcal{I}(X)=\sum_{k=1}^n A_k X A_k^\dagger,
     \end{equation*}
     where $n\leq Kdd_q$. Then, we have that
     \begin{equation*}
         F=\mathcal{I}(I)=\sum_{k=1}^n A_kA_k^\dagger.
     \end{equation*}
     For some $|v\rangle\in \ker(F)$, we have that
     \[
         0=\langle v|F|v\rangle =\sum_{k=1}^n\|A^\dagger_k v\|^2
    \]
    and so
    $
         A_k^\dagger|v\rangle = 0,
     $
      which implies the lemma.
\end{proof}
This result implies that while general quantum instruments can map the maximally mixed state to some \textit{a priori} arbitrary state $\tau$, they will never actually map anything outside the support of $\tau$.
\section{General Public-Coin Lower Bound}\label{sec:general_public_lb}
Here, we prove \cref{thm:public_coin_lb}. To do so, we first lay out three key lemmas, and show how they imply the theorem. We afterwards prove the three lemmas. To recap notation, recall that below we set $d_q=2^{n_q}$, and $K=2^{n_c}$.
\begin{lemma}\label{lemma:instrument_T_bound}
    We use the hard instance in \cref{def:hard_instances}. Let $d\geq 6$, $0<\eps<C$, and $K,d_q\geq 1$. For each distributed node, let $\mathcal{I}_i:\mathbb{C}^{d\times d}\to \mathbb{C}^{d_q\times d_q}\otimes \{0,1\}^{n_c\times n_c}$ be a collection of quantum instruments. For each quantum channel, we let $\tau_i=\mathcal{I}_i(I/d)$, the mapping of the maximally mixed state. We define the augmented map $\Lambda_i(X)=\tau^{-1/4}\mathcal{I}_i(X)\tau^{-1/4}$, and $T_\Lambda=\frac{1}{m}\sum_{i=1}^mM_{\Lambda_i}^\dagger M_{\Lambda_i}$, where $M_{\Lambda_i}$ is the Liouville matrix representation of $\Lambda_i$. Then,
    \begin{equation*}
        \QChi{\mathbb{E}_z\left[\bigotimes_{i=1}^m\mathcal{I}_i(\rho_z)\right]}{\bigotimes_{i=1}^m\mathcal{I}_i(I/d)} < \frac{1}{16},
    \end{equation*}
    unless
    \begin{equation*}
        m\geq \Omega\left(\frac{d\ell}{\eps^2}\cdot\frac{1}{\|\mathcal{V}^\dagger T_\Lambda\mathcal{V}\|_2}\right)
    \end{equation*}
\end{lemma}
\begin{lemma}\label{lemma:instrument_infty_bound}
    For the parameters defined above and a quantum instrument $\mathcal{I}$, we have that:
    \begin{equation*}
        \|M_{\Lambda}\|_\infty\leq \sqrt{d}.
    \end{equation*}
\end{lemma}
\begin{lemma}\label{lemma:instrument_2_bound}
    For the parameters defined above and a quantum instrument $\mathcal{I}$, we have that:
    \begin{equation*}
        \|M_{\Lambda}\|_2\leq d_q\sqrt{dK}
    \end{equation*}
\end{lemma}
\begin{remark}
    We comment on some special cases of the three results above, since they are stated in the most general form possible. First, assume that the quantum instruments are actually quantum channels. That is, $n_c=0$, and thus $K=1$. Then, the only result that changes is \cref{lemma:instrument_2_bound}, which becomes $d_q\sqrt{d}$ instead. One might think it strange that the other two results do not change, but this is not that surprising. First, \cref{lemma:instrument_T_bound} will simply have another operator instead of $T_\Lambda$ that instead of being parametrized by quantum instruments, are parametrized by quantum channels. One can specialize our proof to quantum channels, and verify that this does not lead to any quantitative gain in the conclusion.  Second, it intuitively makes sense that the addition of the classical channels does not change \cref{lemma:instrument_infty_bound}. This is because whilst the classical channels add many more orthogonal directions, they do not actually impact the magnitude of any one direction, meaning they should not impact the infinity norm. On the other hand, it makes sense that the $2$-norm would increase by some factor, since it will have contributions from all of these additional orthogonal directions.

    We also point out that one can also recover, from our method, the results of Doosti \etal~for mixedness preserving channels. Indeed, for a mixedness-preserving channel, we have (abusing the dimensionality of $I$) that  $\tau=I/d_q$. Then, our augmented map is given by $\Lambda_i(X) \eqdef \sqrt{d_q}\Phi_i(X)$ and, following the proof that we provide later, one obtains the lower bounds (restricted to mixedness-preserving channels) of \cite{doosti2026distributedquantumpropertytesting}. 
\end{remark}
We now give the main lower bounding result and its proof, assuming the three lemmas above.
\begin{theorem}[\cref{thm:public_coin_lb}, restated]
    \label{thm:public_coin_lb:detailed}
    Let $d\geq 6$, $0<\eps< \eps_0$, for a sufficiently small absolute constant $\eps_0 >0$, let $K=2^{n_c}$. Then consider the public-coin model, where each distributed node receives a state $\rho$ and communicates a $d_q$-dimensional qudit and a length $n_c$ binary string. To $\eps$-certify a $d$-dimensional state, any successful protocol must have $m=\Omega\left(\frac{d^2}{d_q\sqrt{K}\eps^2}\right)$ nodes, that is,
    \begin{equation*}
        m=\Omega\left(\frac{d^2}{2^{n_q+n_c/2}\eps^2}\right).
    \end{equation*}
\end{theorem}
\begin{proof}
    By \cref{lemma:instrument_minimax_maximin_bounds}, for all choices of bases $\mathcal{V}$, the number of distributed nodes $m$ must be large enough such that
    \begin{equation*}
        \min_{D\in\mathcal{D}_\eps(I_d/d)}\max_{\mathcal{I}_1,...,\mathcal{I}_m}\QChi{\mathbb{E}_{\rho\sim D}\left[\bigotimes_{i=1}^m\mathcal{I}_i(\rho)\right]}{\bigotimes_{i=1}^m\mathcal{I}_i(I_d/d)}\geq\frac{1}{16}.
    \end{equation*}
    By \cref{lemma:instrument_T_bound}, this means we must have
    \begin{equation*}
        m\geq \Omega\left(\frac{d\ell}{\eps^2}\cdot\frac{1}{\|\mathcal{V}^\dagger T_\Lambda\mathcal{V}\|_2}\right).
    \end{equation*}
    Notice that for every choice of basis $V_1,...,V_{d^2-1}$, $\mathcal{V}$ is an isometry, and thus does not impact norms. Then, we have:
    \begin{align*}
        \|\mathcal{V}^\dagger T_\Lambda\mathcal{V}\|_2&= \|T_\Lambda\|_2\\
        &\leq \frac{1}{m}\sum_{i=1}^m\|M_{\Lambda_i}^\dagger M_{\Lambda_i}\|_2\\
        &\leq \frac{1}{m}\sum_{i=1}^m\|M_{\Lambda_i}^\dagger \|_2\|M_{\Lambda_i}\|_\infty\\
        &\leq d_q\sqrt{dK}\sqrt{d}=d_qd\sqrt K,
    \end{align*}
    where in the second inequality, we used the triangle inequality. In the third, we used Holder's inequality, noting that $\frac{1}{2}+\frac{1}{\infty}=\frac{1}{2}$ by the usual convention. Finally, we used \cref{lemma:instrument_2_bound} and \cref{lemma:instrument_infty_bound} in the last step. Finally, substituting this into \cref{lemma:instrument_T_bound}, and setting $\ell=d^2-1$, we get $m=\Omega\left(\frac{d^2}{d_q\sqrt{K}\eps^2}\right)$, as required. 
\end{proof}
We now aim to prove the three lemmas. To do so, we will need some additional results. First is the Kadison-Schwarz inequality.
\begin{lemma}[Kadison--Schwarz]\label{lemma:Kadison_Schwarz}
    Let $\Gamma$ be a unital completely positive map. Then, we have that
    \begin{equation*}
        \Gamma(X^\dagger X)\succeq\Gamma(X)^\dagger\Gamma(X),
    \end{equation*}
    where $A\preceq B$ if and only if $B-A$ is positive semidefinite. Here, unital means that $\Gamma(I)=I$, abusing notation slightly by leaving out the dimensions.
\end{lemma}
Second is a standard moment bound found in a couple of papers:
\begin{lemma}[{\cite[Lemma 2.7; Claim IV.17; Prop 8.13]{doosti2026distributedquantumpropertytesting, DBLP:journals/tit/AcharyaCLST22, etingof2011introductionrepresentationtheory}}]\label{lemma:exponential_expectation_bound}
    Let $z,z'\sim\{-1,1\}^\ell,~\lambda\in\mathbb{R},~A\in\mathbb{R}^{\ell\times\ell}$. Then,
    \begin{equation*}
        \mathbb{E}_{z,z'}[\exp(\lambda z^\top Az')]\leq \exp(C\lambda^2\|A\|_2^2),
    \end{equation*}
    whenever $\lambda\leq \frac{1}{2\|A\|_\infty}$, $C>0$ is some absolute constant.
\end{lemma}
\noindent We start by proving \cref{lemma:instrument_T_bound}.
\begin{proof}[Proof of~\cref{lemma:instrument_T_bound}]
    For ease of notation, let
    \begin{equation*}
        P=\bigotimes_{i=1}^m\mathcal{I}_i(\rho_z),~Q=\bigotimes_{i=1}^m\mathcal{I}_i(I/d)=\bigotimes_{i=1}^m\tau_i.
    \end{equation*}
    Then, we have that 
    \begin{equation*}
        \QChi{P}{Q}=\tr[(Q^{-1/4}(P-Q)Q^{-1/4})^2].
    \end{equation*}
    By \cref{lemma:quantum_ingster-suslina_symm}, we have that
    \begin{equation*}
        1+\QChi{P}{Q}=\mathbb{E}_{z,z'}\left[\prod_{i=1}^m\tr[\tau_i^{-1/2}\mathcal{I}_i(\rho_z)\tau_i^{-1/2}\mathcal{I}_i(\rho_{z'})]\right].
    \end{equation*}
    Let $\mathcal{I}_i(\rho_z)=\tau_i+\mathcal{I}_i(\bar\Delta_z)$. Then, we use the definition of $Z_i$ to get that
    \begin{align*}
        Z_i(z,z')&=\tr[\tau_{i}^{-1/2}\mathcal{I}_i(\bar\Delta_z)\tau_{i}^{-1/2}\mathcal{I}_i(\bar\Delta_{z'})]\\
        &= \tr[\tau_{i}^{-1/4}\mathcal{I}_i(\bar\Delta_z)\tau_{i}^{-1/4}\tau_{i}^{-1/4}\mathcal{I}_i(\bar\Delta_{z'})\tau_{i}^{-1/4}]\\
        &=\InnerNamed{\Lambda_i(\bar\Delta_z)^\dagger}{\Lambda_i(\bar\Delta_{z'})}{HS}\\
        &=\tr[\Lambda_i(\bar\Delta_z)\Lambda_i(\bar\Delta_{z'})],
    \end{align*}
    where the third equality comes from the definition of the Hilbert-Schmidt inner product, and we can remove the $\dagger$ since all matrices are Hermitian.
    Having obtained the expression $Z_i(z,z')$. We then use the definition 
    \begin{equation*}
        \bar\Delta_z=\frac{c\eps N_z}{\sqrt{d\ell}}\sum_{i=1}^\ell z_iV_i,
    \end{equation*}
    and apply linearity of the quantum instruments to get that
    \begin{align*}
        \mathcal{I}_i(\bar\Delta_z)&=\frac{c\eps N_z}{\sqrt{d\ell}}\mathcal{I}_i\left(\sum_{i=1}^\ell z_iV_i\right)
        =\frac{c\eps N_z}{\sqrt{d\ell}}M_{\mathcal{I}_i}\vecmap\left(\sum_{i=1}^\ell V_iz_i\right)\,,
    \end{align*}
    so that
    \begin{align*}
        \mathcal{I}(\bar\Delta_z)&= \frac{c\eps N_z}{\sqrt{d\ell}}M_\mathcal{I}\mathcal{V}z.
    \end{align*}
    This implies that
    \begin{align*}
        Z_i(z,z')&=\frac{c^2\eps^2N_zN_{z'}}{d\ell}z^\top\mathcal{V}^\dagger M_{\mathcal{I}_i}^\dagger M_{\mathcal{I}_i}\mathcal{V}z'\\
        \sum_{i=1}^mZ_i(z,z')&= \frac{mc^2\eps^2N_zN_{z'}}{d\ell}z^\top\mathcal{V}^\dagger T_\mathcal{I}\mathcal{V}z'.
    \end{align*}
    Hence, letting $A=\mathcal{V}^\dagger T_\mathcal{I}\mathcal{V}$, we have:
    \begin{equation*}
        1+\QChi{P}{S}\leq \mathbb{E}_{z,z'}\left[\exp\left( \frac{mc^2\eps^2N_zN_{z'}}{d\ell}z^\top Az'\right)\right].
    \end{equation*}
    Using the same results as Liu \etal on $N_z$, we find that it equals $1$, except with exponentially small probability \cite{pmlr-v247-liu24a}. This then gives that
    \begin{equation*}
        1+\QChi{P}{S}\leq \mathbb{E}_{z,z'}\left[\exp\left( \frac{mc^2\eps^2}{d\ell}z^\top Az'\right)\right]+\frac{4}{e^d}.
    \end{equation*}
    Then, setting $\lambda=\frac{mc^2\eps^2}{d\ell}>0$, we can analyze two cases. First, if $\lambda>\frac{1}{2\|A\|_\infty}$, we have that
    \begin{align*}
        m&>\frac{d\ell}{2c^2\eps^2\|A\|_\infty}
        \geq\frac{d\ell}{2c^2\eps^2\|A\|_2}
        =\Omega\left(\frac{d\ell}{\eps^2\|A\|_2}\right),
    \end{align*}
    which is the required bound. 
    Otherwise, we have $\lambda\leq\frac{1}{2\|A\|_\infty}$, and then, by \cref{lemma:exponential_expectation_bound}, we get
    \begin{align*}
        \mathbb{E}_{z,z'}\left[\exp\left(\lambda z^\top Az'\right)\right]&\leq \exp(C\lambda^2\|A\|_2)
    \end{align*}
    so that
    \begin{align*}
        1+\QChi{P}{S}&\leq\exp\left(C\frac{m^2\eps^4}{d^2\ell^2}\|\mathcal{V}^\dagger T_\Psi\mathcal{V}\|_2^2\right)+4e^{-d}.
    \end{align*}
    Now, for $d\geq 6$, $4e^{-d}=o(1)$, and so $\QChi{P}{S}\leq 1/16$, unless we have that
    $
        C\frac{m^2\eps^4}{d^2\ell^2}\|\mathcal{V}^\dagger T_\Psi\mathcal{V}\|_2^2=\Omega(1)
    $, or, equivalently,
    \[
        m\geq\Omega\left(\frac{d\ell}{\eps^2}\cdot\frac{1}{\|\mathcal{V}^\dagger T_\Phi\mathcal{V}\|_2}\right),
    \]
    which is the desired bound. Hence, in both cases we get
    \begin{equation*}
        m\geq \Omega\left(\frac{d^3}{Kd_q^2\eps^2}\right),
    \end{equation*}
    as required.
\end{proof}
Next, we prove \cref{lemma:instrument_2_bound}.
\begin{proof}[Proof of~\cref{lemma:instrument_2_bound}]
    We first show $\|F^{-1/4}\mathcal{I}(X)F^{-1/4}\|_2\leq \|X\|_2$. Set $F\eqdef \mathcal{I}(I)=d\tau$, and $\Gamma(X)\eqdef F^{-1/2}\mathcal{I}(X)F^{-1/2}$. Note this is completely positive, unital, and that $F$ is also Hermitian (since it is positive semidefinite). Hence, by Kadison--Schwarz, we get that
    $\Gamma(X)^\dagger\Gamma(X)\preceq \Gamma(X^\dagger X)$, that is, 
    \begin{align*}
       F^{-1/2}\mathcal{I}(X)^\dagger F^{-1}\mathcal{I}(X)F^{-1/2}&\preceq F^{-1/2}\mathcal{I}(X^\dagger X)F^{-1/2}\\
       \mathcal{I}(X)^\dagger F^{-1}\mathcal{I}(X)&\preceq \mathcal{I}(X^\dagger X).
   \end{align*}
   By the same argument, we also get the other direction:
   \begin{equation*}
       \mathcal{I}(X) F^{-1}\mathcal{I}(X)^\dagger\preceq \mathcal{I}(X X^\dagger).
   \end{equation*}
   We can then take traces of both sides, and use the fact that $\Phi$ is trace preserving, to get that
    \begin{align*}
        \tr[\mathcal{I}(X)^\dagger F^{-1}\Phi(X)]&\leq \tr[\mathcal{I}(X^\dagger X)]
        =\tr[X^\dagger X]
        =\|X\|_2^2.
    \end{align*}
    A similar argument yields that
    \begin{equation*}
        \tr[\mathcal{I}(X)F^{-1}\mathcal{I}(X)^\dagger]\leq \|X\|_2^2.
    \end{equation*}
   We can decompose $F$ into its eigenbasis $\{\lambda_a\}$, with eigenvectors $|a\rangle$ respectively. Then, we have that, for every $i,j$,
    \begin{align*}
    \langle i|F^{-1/4}\mathcal{I}(X)F^{-1/4}|j\rangle &=\sum_{k,\ell}\lambda_k^{-1/4}\lambda_\ell^{-1/4}\langle i|k\rangle \langle k|\mathcal{I}(X)|\ell\rangle\langle \ell|j\rangle\\
    &=\lambda_i^{-1/4}\lambda_j^{-1/4}\langle i|\mathcal{I}(X)|j\rangle\\
    &=\lambda_i^{-1/4}\lambda_j^{-1/4}\mathcal{I}_{i,j}(X)
    \end{align*}
    which implies
    \begin{align*}
    \|\langle i|F^{-1/4}\mathcal{I}(X)F^{-1/4}|j\rangle\|^2&=\lambda_i^{-1/2}\lambda_j^{-1/2}\|\mathcal{I}_{i,j}(X)\|^2
    \end{align*}
    and so 
    \begin{align*}
    \|F^{-1/4}\mathcal{I}(X)F^{-1/4}\|_2^2&=\sum_{i,j}\frac{\|\mathcal{I}_{i,j}(X)\|^2}{\sqrt{\lambda_i\lambda_j}}.
    \end{align*}
    By the AM-GM inequality, we have that $\frac{1}{\sqrt{\lambda_i \lambda_j}}\leq \frac{1}{2}(\frac{1}{\lambda_i}+\frac{1}{\lambda_j})$, so
    \begin{align*}
        \|F^{-1/4}\mathcal{I}(X)F^{-1/4}\|_2^2&\leq \frac{1}{2}\sum_{i,j}\left(\frac{1}{\lambda_i}+\frac{1}{\lambda_j}\right)\|\mathcal{I}_{i,j}(X)\|^2\\
        &=\frac{1}{2}\sum_{i,j}\frac{1}{\lambda_i}\|\mathcal{I}_{i,j}(X)\|^2+\frac{1}{2}\sum_{i,j}\frac{1}{\lambda_j}\|\mathcal{I}_{i,j}(X)\|^2.
    \end{align*}
    Considering the first term in the sum:
    \begin{align*}
        \frac{1}{2}\sum_{i,j}\frac{1}{\lambda_i}\|\mathcal{I}_{i,j}(X)\|^2&=\frac{1}{2}\sum_i\frac{1}{\lambda_i}\sum_j\|\mathcal{I}_{i,j}(X)\|^2.
    \end{align*}
    We next note that for any matrix $B$, we have that $(BB^\dagger)_{ac}=\sum_bB_{ab}B^\dagger_{bc}$. Letting $c=a$, we have that $(BB^\dagger)_{aa}=\sum_b\|B_{ab}\|^2$. Substituting this in above, we get that
    \begin{align*}
        \frac{1}{2}\sum_i\frac{1}{\lambda_i}\sum_j\|\mathcal{I}_{i,j}(X)\|^2&=\frac{1}{2}\sum_i\frac{1}{\lambda_i}(\mathcal{I}(X)\mathcal{I}(X)^\dagger)_{ii}
        =\frac{1}{2}\tr[F^{-1}\mathcal{I}(X)\mathcal{I}(X)^\dagger]
        =\frac{1}{2}\tr[\mathcal{I}(X)^\dagger F^{-1}\mathcal{I}(X)]\,.
    \end{align*}
    Similarly, for the second term, we have that:
    \begin{equation*}
        \frac{1}{2}\sum_{i,j}\frac{1}{\lambda_j}\|\mathcal{I}_{i,j}(X)\|^2=\frac{1}{2}\tr[\mathcal{I}(X)F^{-1}\mathcal{I}(X)^\dagger].
    \end{equation*}
    Putting both together, we get that
    \begin{equation*}
        \|F^{-1/4}\mathcal{I}(X)F^{-1/4}\|_2^2\leq \frac{1}{2}\tr[\mathcal{I}(X)^\dagger F^{-1}\mathcal{I}(x)]+\frac{1}{2}\tr[\mathcal{I}(X)F^{-1}\mathcal{I}(X)^\dagger].
    \end{equation*}
    But recall the right hand side now is upper bounded by $\frac{1}{2}\|X\|_2^2+\frac{1}{2}\|X\|_2^2$. Hence, this gives that $\|F^{-1/4}\mathcal{I}(X)F^{-1/4}\|_2^2\leq \frac{1}{2}2\|X\|_2^2=\|X\|^2_2$.
    Then, as a map, this implies that
    \begin{equation*}
        \sup_{\|X\|_2\neq 0}\frac{\|F^{-1/4}\mathcal{I}(X)F^{-1/4}\|_2}{\|X\|_2}\leq 1.
    \end{equation*}
    Further, using that $\tau=F/d$, we get that
    \begin{equation*}
        \sup_{\|X\|_2\neq 0}\frac{\|\Lambda(X)\|_2}{\|X\|_2}\leq \sqrt d.
    \end{equation*}
    Then, recalling the Liouville representation of the numerator $M_\mathcal{I}$, we have that
    \begin{equation*}
        \|M_{\Lambda}\|_\infty=\sup_{\|\vecmap(X)\|_2=1}\|M_{\Lambda}\vecmap(X)\|_2.
    \end{equation*}
    Choosing $X$ such that $\|X\|_2=1$, this implies that $\|\vecmap(X)\|_2=1$, and thus, $\|\Lambda(X)\|_2=\|M_\Lambda\|_\infty\leq\sqrt d$, as required.
\end{proof}
We are left with the last building block, \cref{lemma:instrument_2_bound}.
\begin{proof}[Proof of~\cref{lemma:instrument_2_bound}]
    We start from the above, where $\|M_\Lambda\|_\infty\leq \sqrt{d}$ implies that the singular values are at most $\sqrt{d}$. Now, we defined $\tau=\mathcal{I}(I/d)$. Then, by the definition of the instrument, we have that
    \begin{equation*}
        \tau=\mathcal{I}(I/d)=\sum_{a=1}^K\tau_a\otimes|a\rangle\langle a|,
    \end{equation*}
    where $\tau_a$ = $\mathcal{E}_a(I/d)$, the individual maps comprising the instrument. Also, we know that
    \begin{equation*}
        \|M_\Lambda\|_2^2=\sum_j s_j(M_\Lambda)^2\leq \rank(M_\Lambda{I})\|M_\Lambda\|_\infty^2,
    \end{equation*}
    where $s_j$ denotes the $j$-th singular value. Let $r_a=\rank(\tau_a)$. Recall since this is a $d_q$-dimensional qudit, we have $r_a\leq d_q$. Then, this implies the rank of the operator $\mathcal{E}$ is at most $d_q^2$. Furthermore, there are $K$ basis vectors for the classical portion of the quantum channel, and hence $\rank(M_\Lambda)\leq Kd_q^2$. Then, putting things together, we get that 
    \begin{align*}
        \|M_\Lambda\|_2^2&\leq d_q^2dK
    \end{align*}
    that is $\|M_\Lambda\|_2\leq d_q\sqrt{dK}$,
    as required.
\end{proof}
\section{General Private-Coin Lower Bound}
Having established the three lemmas in the quantum instrument case we have all the tools needed to prove the private-coin lower bound as well, \cref{thm:private_coin_lb}, restated below:
\begin{theorem}[\cref{thm:private_coin_lb}, restated]\label{thm:private_coin_LB}
    In the same parameter setting as~\cref{thm:public_coin_lb:detailed} except in the private-coin case, any private-coin protocol for $\eps$-certifying $d$-dimensional states requires $m\geq\Omega\left(\frac{d^3}{Kd_q^2\eps^2}\right)$ nodes; that is,
    \begin{equation*}
        m\geq\Omega\left(\frac{d^3}{2^{n_c+2n_q}\eps^2}\right).
    \end{equation*}
\end{theorem}
\begin{proof}
    Fix quantum instruments $\mathcal{I}_1,\dots,\mathcal{I}_m$. It suffices to choose $\mathcal{V}$ to minimize $\|\mathcal{V}^\dagger T_\Lambda \mathcal{V}\|_2$, for fixed channels. We start by restricting to the traceless subspace, $H\eqdef \{X=X^\dagger\mid \tr X=0\}$, which has dimension $d^2-1$. Let $P$ be the projector onto $H$, and consider the positive semidefinite operator on this subspace, $B\eqdef PT_\Lambda P$. Denote its eigenvalues by $0\leq \mu_1\leq...\leq \mu_N$. Let $V_1,...,V_\ell$ be the eigenvectors corresponding to the $\ell$ smallest eigenvalues, for $\ell$ to be chosen momentarily. Then, with this choice, we have that $\mathcal{V}^\dagger T_\Lambda\mathcal{V}$ is a diagonal matrix with entries $\mu_1,...,\mu_\ell$ along its diagonal, and hence
    \begin{equation*}
        \|\mathcal{V}^\dagger T_\Lambda\mathcal{V}\|_2^2=\sum_{i=1}^\ell \mu_i^2.
    \end{equation*}
    Each of the $\ell$ smallest eigenvalues are at most the average of the remaining eigenvalues, so
    \begin{align*}
        \|\mathcal{V}^\dagger T_\Lambda\mathcal{V}\|_2&\leq \sqrt{\ell}\frac{\sum_{i=\ell+1}^{d^2-1}\mu_i}{d^2-\ell-1}\\
        &\leq \sqrt{\ell}\frac{\tr[T_\Lambda]}{d^2-l-1}\\
        &=\frac{\sqrt{\ell}}{d^2-\ell-1}\frac{\sum_{i=1}^m\|M_{\Lambda_i}\|_2^2}{m}\\
        &\leq \frac{\sqrt{\ell}}{d^2-\ell-1}dd_q^2K.
    \end{align*}
    We choose $\ell\eqdef d^2/2$, which gives that $\|\mathcal{V}^\dagger T_\Lambda\mathcal{V}\|_2\leq O(d_q^2 K)$. Plugging this in to \cref{lemma:instrument_T_bound}, we get that
    \begin{equation*}
        m\geq \Omega\left(\frac{d^3}{Kd_q^2\eps^2}\right),
    \end{equation*}
    as required.
\end{proof}
\section{Private-Coin Upper Bound}
In this section, we prove \cref{thm:private_coin_no_classical_ub} by giving a private-coin algorithm with $n_c=0$ (i.e., no classical communication), which in view of~\cref{thm:private_coin_no_classical_lb} is tight up to a logarithmic factor. We assume that $d_q|d$, and as such, $\mathbb{C}^d\cong A\otimes B$, where $\dim(A)=d/d_q$, and $\dim(B)=d_q$. (This is without loss of generality, by a simple padding argument.) We begin with the following key result, which we prove later, and will underlie our algorithm.
\begin{theorem}\label{thm:good_unitaries}
    There exists unitaries $U_1,...,U_L$, where $L=O\left(\frac{d^2}{d_q^2}\log d\right)$, such that for every traceless Hermitian $X\in\mathbb{C}^{d\times d}$,
    \begin{equation*}
        \frac{1}{L}\sum_{i=1}^L\|\tr_B(U_iXU_i^\dagger)\|_2^2\geq \frac{1}{4}\cdot \frac{d_q}{d}\|X\|_2^2.
    \end{equation*}
\end{theorem}
For the remainder of this section, for shorthand, we will refer Haar-random unitaries simply as unitaries. We also use as a black box subroutine, the Hilbert--Schmidt estimator of \cite{DBLP:conf/stoc/BadescuO019}:
\begin{lemma}[{Hilbert--Schmidt Estimator, see \cite[Prop.~5.6]{DBLP:conf/stoc/BadescuO019}}]\label{lemma:hilbert_schmidt_estimator}
    Given $t$ copies of $\rho$ and $\sigma$, there is an unbiased estimator for $D_{HS}^2(\rho,\sigma)=\tr[(\rho-\sigma)^2]$ with variance bounded by
    \begin{equation*}
        O\left(\frac{1}{t^2}+\frac{D_{HS}^2(\rho,\sigma)}{t}\right).
    \end{equation*}
\end{lemma}
This lemma is only implicit in \cite{DBLP:conf/stoc/BadescuO019}, where it is used as a subroutine in their tester for Hilbert--Schmidt distance. We will refer to this estimator as the BOW estimator. As a final piece of notation, for any unitary $U$, let $\Phi_U(X)$ be the quantum channel such that $\Phi(U)(X)\eqdef \tr_B(UXU^\dagger)$. Our algorithm for the private-coin setting is provided in \cref{alg:private_coin_no_classical}.
\begin{algorithm}[H]
    \caption{Private-Coin Distributed Algorithm for Certification}\label{alg:private_coin_no_classical}
    \begin{algorithmic}[1]
        \State Set $L=O(\frac{d^2}{d_q^2}\log d)$, $t=O(\frac{d}{\eps^2})$, such that $m=Lt$. Set $\beta=0.25\frac{d_q\eps^2}{d^2}$.
        \State \Comment{All nodes have previously agreed on a good set of unitaries $U_1,...,U_L$ satisfying \cref{thm:good_unitaries}, and each node is assigned a unique label of the form $N_{r,a}$, for $r\in [L]$ and $a\in[t]$.}
        \State Each distributed node $N_{r,a}$ receives a copy $\rho$ and sends $\omega_r=\Phi_{U_r}(\rho)$ to the central node.
        \State The central node receives $t$ copies of $\omega_r$ for each unitary $U_r$. 
        \State With full knowledge of $\sigma$, the central node prepares $t$ copies of $\tau_r=\Phi_{U_r}(\sigma)$.
            \State The central node computes an estimate $\hat{D}$ of $D=\frac{1}{L}\sum_{r=1}^L\|\omega_r-\tau_r\|_2^2$, using the Hilbert--Schmidt estimator of \cref{lemma:hilbert_schmidt_estimator} and $t$ copies of $\omega_r$ and $\tau_r$ for each estimate of $\|\omega_r-\tau_r\|_2^2$.
            \If{ $\hat{D}>\beta/2$ }
            \State \Return \reject
            \Else \State \Return \accept
            \EndIf
    \end{algorithmic}
\end{algorithm}
Assuming \cref{thm:good_unitaries}, we analyze \cref{alg:private_coin_no_classical}, thus establishing \cref{thm:private_coin_no_classical_ub}.
\begin{proof}[Proof of~\cref{thm:private_coin_no_classical_ub}]
    For ease of notation, let $\Delta\eqdef \rho-\sigma$. The central node wishes to estimate the quantity
    \begin{equation}\label{eq:estimated_qty}
        D\eqdef\frac{1}{L}\sum_{r=1}^L\|\Phi_{U_r}(\rho)-\Phi_{U_r}(\sigma)\|_2^2=\frac{1}{L}\sum_{r=1}^L\|\Phi_{U_r}(\Delta)\|_2^2\,,
    \end{equation}
    the last equality by linearity of channels. There are two cases:
    \begin{itemize}
        \item if $\sigma=\rho$, then, for each unitary $U_r$, we have that $\|\Phi_{U_r}(\rho)-\Phi_{U_r}(\sigma)\|_2^2=0$, and hence $D=0$;
        \item if $\|\rho-\sigma\|_1\geq \eps$, by Cauchy--Schwarz $\|\Delta\|_2\geq\frac{\|\Delta\|_1}{\sqrt{d}}\geq\frac{\eps}{\sqrt{d}}$. Then, applying \cref{thm:good_unitaries}, we have that 
    \begin{equation*}
        D\geq \frac{1}{4}\cdot\frac{d_q\eps^2}{d^2}=\beta.
    \end{equation*}
    \end{itemize}
    In order to perform $\eps-$certification, it is thus enough to distinguishing between $D=0$ and $D\geq \beta$. Using $t$ copies for each $1\leq r\leq L$, we obtain an estimate of $z_r \eqdef \|\Phi_{U_r}(\Delta)\|_2^2$. This yields an estimate $\hat{D}=\frac{1}{L}\sum_{r=1}^L z_r$, by taking the sum of all the individual estimates. Since the estimate $z_r$ are independent, and since, for each $1\leq r\leq L$,
    \begin{equation*}
        \Var[z_r]\leq C\left(\frac{1}{t^2}+\frac{\mathbb{E}[z_r]}{t}\right),
    \end{equation*}
    for some absolute constant $C>0$, we get
    \begin{align*}
        \Var[\hat D]&\leq \frac{1}{L^2}\sum_{r=1}^L\Var[z_r]
        \leq C\left(\frac{1}{Lt^2}+\frac{D}{Lt}\right),
    \end{align*}
   with the last equality using that  $\mathbb{E}[z_r]=\|\Phi_{U_r}(\Delta)\|_2^2$. Next, we consider the two cases. First, if $\rho=\sigma$, then we have that
    \begin{equation*}
        \Var[\hat D]\leq \frac{C}{Lt^2}.
    \end{equation*}
    Then, by Chebyshev's inequality, we have that
    \begin{align*}
        \Pr[\hat D\geq \beta/2]&\leq O\left(\frac{\Var[\hat D]}{\beta^2}\right)
        =O\left(\frac{1}{\log d}\right),
    \end{align*}
    which is at most $1/3$ for large enough $d$. Else, if $\|\rho-\sigma\|_1\geq\eps$, then we have $D\geq \beta$, so that $D-\beta/2\geq D/2$, and again applying Chebyshev's, we get that
    \begin{align*}
        \Pr]\hat D<\beta/2]&\leq \Pr[|\hat D-D|\geq D/2]\\
        &\leq O\left(\frac{\Var[\hat D]}{D^2}\right)\\
        &\leq O\left(\frac{1}{LtD}+\frac{1}{Lt^2D^2}\right)\\
        &\leq O\left(\frac{1}{Lt\beta}+\frac{1}{Lt^2\beta^2}\right)\\
        &=O\left(\frac{d_q}{d\log d}+\frac{1}{\log d}\right)\\
        &\leq O\left(\frac{1}{\log d}\right).
    \end{align*}
    again at most $1/3$ for large enough $d$. In doing so, with probability $2/3$, \cref{alg:private_coin_no_classical} distinguishes $D=0$ from $D\geq \beta/2$, and hence can $\eps$-certify with probability of success $2/3$. The number of nodes used is 
    \begin{equation*}
        m=Lt=O\left(\frac{d^2}{d_q^2}\log d\frac{d}{\eps^2}\right)=O\left(\frac{d^3}{d_q^2\eps^2}\log d\right).
    \end{equation*}
    Hence, it suffices to take $m= O\left(\frac{d^3}{d_q^2\eps^2}\log d\right)$.
\end{proof}
\subsection{The Proof of \cref{thm:good_unitaries}}
The only thing that remains to be proven is the existence of a (small) set of good unitaries satisfying our requirements. We will use the probabilistic method, showing that a set of Haar-random unitaries satisfies this with positive probability, and hence prove \cref{thm:good_unitaries}. Before going into the proof, we establish some notation and some intermediate results we will use. First, let $H=\{X=X^\dagger, \tr X=0\}$ be the space of traceless Hermitian operators. Since $X$ has dimension $d^2$, $H$ has dimension $d^2-1$. Next, let $W=\{Y\otimes I_B:Y=Y^\dagger,~\tr Y=0\}\subseteq H$ be the subset of traceless Hermitian operators that are identity on the space $B$. Since $A$ has dimension $d_q^2$, this space has dimension $d_q^2-1$. Let $\Pi_W$ be the orthogonal projector onto $W$. Next, for any unitary $U$, let $W_U=U^\dagger W U$, where abusing notation, be the space of operators obtained by taking any operator in $W$ sandwiched between $U^\dagger$ and $U$. Let $\Pi_U$ be the orthogonal projector onto $W_U$. We prove the following lemma first:
\begin{lemma}\label{lemma:W_projector}
    For any $X\in H$, we have that $\|\tr_B[X]\|_2^2=\frac{d}{d_q}\|\Pi_WX\|_2^2$.
\end{lemma}
\begin{proof}
    Since $X$ is traceless, we have that $\tr_A(\tr_B(X))=\tr_{AB}(X)=0$. We specifically denote under which spaces the traces are over. Next, we can consider the projection, which is $\Pi_W(X)=\frac{1}{b}\tr_B(X)\otimes I_B\in W$, noting that because of the above, this is still traceless. Taking norms, we now get
    \begin{align*}
        \|\Pi_W(X)\|_2^2&=\||\frac{1}{b}\tr_B(X)\otimes I_B\|_2^2
        =\frac{1}{b}\|\tr_B(X)\|_2^2\,,
    \end{align*}
    as claimed.
\end{proof}
Next, we analyze how the action of a channel $\Phi_U$ impacts the norm. Specifically, we show:
\begin{lemma}\label{lemma:channel_norm}
    For any unitary $U$,
    \begin{equation*}
        \|\Phi_U(X)\|_2^2=\frac{d}{d_q}\|\Pi_UX\|_2^2.
    \end{equation*}
\end{lemma}
\begin{proof}
    For a unitary $U$, if $R_U$ denotes the conjugation map $R_U(X)=UXU^\dagger$, then the map $\Pi_U=R^{-1}_U\Pi_WR_U$, which can be verified just by calculation. This then means that $\Pi_UX=U^\dagger\Pi_W(UXU^\dagger)U$. Then, taking norms, and using that $U$ and $U^\dagger$ do not change norms, we get that $\|\Pi_UX\|_2=\|\Pi_W(UXU^\dagger)\|_2$. We denote $Z=UXU^\dagger$. Then, recalling that $\Phi_U(X)=\tr_B(UXU^\dagger)=\tr_B(Z)$, we get from \cref{lemma:W_projector} that $\|\Phi_U(X)\|_2^2=|\tr_B[Z]\|_2^2=\frac{d}{d_q}\|\Pi_WZ\|_2^2$. This then gives that $\|\Phi_U(X)\|_2^2=\frac{d}{d_q}\|\Pi_UX\|_2^2$, as required.
\end{proof}
It is also helpful to recall some properties of the projector $\Pi_U$. Namely, since it is an orthogonal projector, it has only eigenvalues $0$ and $1$. Moreover, we have that $0\preceq \Pi_U\preceq I$. Finally, since it projects onto a subspace of dimension $d_q^2-1$, we have that $\tr_H[\Pi_U]=d_q^2-1$. Here, the trace is of $\Pi_U$ as a linear operator, and not of the quantum state.

The above properties all establish something about a particular $U$. However, we wish to sample a list of random matrices, and then say something about their property. Hence, we now wish to derive properties about the expectation of this projector $\Pi_U$. We will make use of Schur's Lemma, a standard result in representation theory, see for instance \cite[Corollary 1.17]{etingof2011introductionrepresentationtheory}:
\begin{lemma}[Schur's Lemma]\label{lemma:Schur}
    Let $V$ be an algebraically closed field. Then, endomorphisms $\phi\colon V\to V$ of finite dimensional irreducible representations of $V$ are scalar multiples of the identity operator.
\end{lemma}
We make a few comments on the above. First, for completeness, it is stated more generally than we needed. Indeed, we are trivially in an algebraically closed field, and all endomorphisms will also trivially be finite dimensional, just by our setting.  Furthermore, it is worth pointing out again that the identity operator is not the identity matrix. $I_H$ is the identity map that sends an element of $H$ to itself. We use one final fact of traceless Hermitian matrices:
\begin{fact}\label{lemma:traceless_hermitian_operators}
    The conjugation action of traceless Hermitian matrices is irreducible.
\end{fact}
To see this, first note that the adjoints are ideals of the traceless matrices \cite{AlAssal2014InvitationLieAlgebras}, which is a simple space, and hence irreducible \cite{Morgan2022SimpleLieAlgebrasCompactLieGroups}. We can now state and prove the main lemma:
\begin{lemma}\label{lemma:projector_expectation}
    \begin{equation*}
        \mathbb{E}_U[\Pi_U]=\frac{d_q^2-1}{d^2-1}I_{H}.
    \end{equation*}
\end{lemma}
\begin{proof}
    For shorthand, let $T=\mathbb{E}_U[\Pi_U]$. For any fixed Haar unitary matrix $V$, let $R_V(X)=VXV^\dagger$. Notice that
    \begin{align*}
        R_V W_U&=V(U^\dagger W U)V^\dagger\\
        &=(UV^\dagger)^\dagger W(UV^\dagger)\\
        &=W_{UV^\dagger}.
    \end{align*}
    This implies that $R_V\Pi_UR_V^{-1}=\Pi_{UV^\dagger}$. Notice that since $U$ and $V$ are Haar-random matrices, then $UV^\dagger$ is also a Haar-random matrix. Hence,
    \begin{align*}
        R_VTR_V^{-1}&=R_V\mathbb{E}_U[\Pi_U]R_V^{-1}\\
        &=\mathbb{E}_U[R_V\Pi_UR_V^{-1}]\\
        &=\mathbb{E}_U[P_{UV^\dagger}]\\
        &=T.
    \end{align*}
    Hence, for all Haar-random unitaries, $R_VT=TR_V$, meaning that $T$ commutes with every conjugation action on $H$. Then, by \cref{lemma:traceless_hermitian_operators}, this is irreducible as a representation. Then, this means we can apply Schur's lemma, and $T$ must be a scalar multiple of the identity, so $T=\lambda I_H$. To find the value of $\lambda$, we take traces over the space $H$. First, notice that 
    \begin{align*}
        \tr_H[T]&=\tr[\mathbb{E}[\Pi_U]]=\mathbb{E}_U[\tr_H[\Pi_U]].
    \end{align*}
    Since each $\Pi_U$ has rank $d_q^2-1$, we have that $\tr_H[T]=d_q^2-1$. On the other hand, we can take
    \begin{equation*}
        \tr_H[T]=\lambda\tr_H[I_H]=\lambda(d^2-1).
    \end{equation*}
    Then, equating the two, we get that 
    \begin{align*}
        \lambda&=\frac{d_q^2-1}{d^2-1}\\
        \mathbb{E}_U[\Pi_U]&=\frac{d_q^2-1}{d^2-1}I_H.
    \end{align*}
\end{proof}
We can now proceed to apply the Matrix Chernoff inequalities to our list of unitaries. We state the Matrix Chernoff inequality first \cite[Corollary 5.2]{Tropp_2011}.
\begin{theorem}[Matrix Chernoff Inequality]\label{thm:matrix_chernoff_inequality}
    Consider a finite sequence $\{X_k\}$ of $d$-dimensional, independent, random, self adjoint matrices such that each satisfies $X_k\succeq 0$ and $\lambda_{\max}(X_k)\leq R$ almost surely, where $\lambda_{\max}$ denotes the larges eigenvalue. Then, letting $\lambda_{\min}$ also denote the smallest eigenvalue, denote
    \begin{equation*}
        \mu_{\min}=\lambda_{\min}\left(\sum_k\mathbb{E}X_k\right),~\mu_{\max}=\lambda_{max}\left(\sum_k\mathbb{E}X_k\right).
    \end{equation*}
    Then, for $0<s<1$,
    \begin{equation*}
        \Pr\left[\lambda_{\min{}}\left(\sum_k X_k\right)\leq (1-s)\mu_{\min{}}\right]\leq d\cdot\exp(-s^2\mu_{\min{}}/2R)
    \end{equation*}
\end{theorem}
Again, we stated the results in more generality than needed. Recall since we will be applying this to projector matrices, we trivially satisfy the required conditions, and we have that $R=1$. We will now aim to prove the following result
\begin{lemma}\label{lemma:deterministic_good_choice}\
    There exist unitaries $U_1,...,U_L$ such that 
    \begin{equation*}
        \frac{1}{L}\sum_{r=1}^L\Pi_{U_r}\succeq\frac{1}{2}\frac{d_q^2-1}{d^2-1}I,
    \end{equation*}
    for $L=O\left(\frac{d^2}{d_q^2}\log d\right)$.
\end{lemma}
\begin{proof}
    Recall by \cref{lemma:projector_expectation}, for each projector we have that $\mathbb{E}[P_{U_r}]=\frac{d_q^2-1}{d^2-1}I$, and hence, 
    \begin{equation*}
        \mathbb{E}\left[\sum_{r=1}^L\Pi_{U_r}\right]=L\frac{d_q^2-1}{d^2-1}I.
    \end{equation*}
    Then, this also implies that $\mu_{\min{}}=L\frac{d_q^2-1}{a^2-1}$, since $I$ only has eigenvalue $1$. Then, we can apply the matrix Chernoff bound with $s=0.5$, which gives that
    \begin{equation*}
        \Pr\left[\lambda_{\min}\left(\sum_{r=1}^L\Pi_{U_r}\right)\leq \frac{1}{2}L\frac{d_q^2-1}{d^2-1}\right]\leq (d^2-1)\exp\left(-\frac{1}{8}L\frac{d_q^2-1}{d^2-1}\right).
    \end{equation*}
    We want to make this probability less than $1$, which implies that we need to choose $L= O\left(\frac{d^2}{d_q^2}\log d\right)$. Then, with positive probability, we obtain a set of $L$ projectors such that the minimum eigenvalue of the sum is large enough, meaning such unitaries $U_1,...,U_L$ with
    \begin{equation*}
        \sum_{r=1}^L\Pi_{U_r}\succeq\frac{1}{2}L\frac{d_q^2-1}{d^2-1}I
    \end{equation*}
    exist.
\end{proof}
We finally have all the ingredients to prove \cref{thm:good_unitaries}.
\begin{proof}[Proof of~\cref{thm:good_unitaries}]
    Take any $X\in H$, and fix a good list of unitaries $U_1,..,U_L$, satisfying \cref{lemma:deterministic_good_choice}. Then, we have that
    \begin{equation*}
        \left\langle X, \left(\sum_{r=1}^L\Pi_{U_r}\right)X\right\rangle\geq \frac{1}{2}L\frac{d_q^2-1}{d^2-1}\langle X,X\rangle.
    \end{equation*}
    The right hand side equals $\frac{1}{2}L\frac{d_q^2-1}{d^2-1}\|X\|_2^2$. Using the fact that $\Pi_{U_r}$ is an orthogonal projection, we have that 
    \begin{align*}
        \langle X, \Pi_{U_r}X\rangle&=\langle \Pi_{U_r}X,\Pi_{U_r}X\rangle=\|\Pi_{U_r}X\|_2^2\\
        \left\langle X, \left(\sum_{r=1}^L\Pi_{U_r}\right)X\right\rangle&=\sum_{r=1}^L \langle X, \Pi_{U_r}X\rangle.
    \end{align*}
    Thus, putting everything together and dividing by $L$, we get that
    \begin{equation*}
        \frac{1}{L}\sum_{r=1}^L\|P_{U_r}X\|_2^2\geq \frac{1}{2}\frac{d_q^2-1}{d^2-1}\|X\|_2^2.
    \end{equation*}
    Then, using \cref{lemma:channel_norm}, we have that
    \begin{equation*}
        \frac{1}{L}\sum_{r=1}^L\|\Phi_{U_r}(X)\|_2^2\geq\frac{d}{d_q}\cdot\frac{1}{2}\cdot\frac{d_q^2-1}{d^2-1}\|X\|_2^2.
    \end{equation*}
    Notice that for $d_q\geq 2$, we have that $q^2-1\geq \frac{1}{2}q^2$, and also $d^2-1\leq d^2$. Thus, rearranging, we get that 
    \begin{equation*}
         \frac{1}{L}\sum_{r=1}^L\|\Phi_{U_r}(X)\|_2^2\geq\frac{1}{4}\cdot\frac{d_q}{d}\|X\|_2^2,
    \end{equation*}
    as required. The final thing to point out is that we required $L$ large enough so that \cref{lemma:deterministic_good_choice} held, meaning we required $L=O\left(\frac{d^2}{d_q^2}\log d\right)$, completing the proof.
\end{proof}
\section{General Public-Coin Upper Bound}
In this section, we give an algorithm with sample complexity matching the lower bound given in \cref{sec:general_public_lb}. Before jumping into the algorithm and the proofs, we make a few observations to provide intuition as to why the setting with both classical and quantum communication allowed is significantly more challenging. For example, one might argue as follows: we know how to test when communicating classical bits only, and we now know how to test when communicating qudits only, so we could just run both tests and accept a state if both pass. More specifically, the central node knows the quantum instruments that the distributed nodes are using to communicate. Hence, they can run those quantum instruments on the known state $\sigma$, and compare it to the classical bits obtained from $\rho$, and the quantum qudit obtained from $\rho$. The issue is that this approach ignores the relations between the joint states, and thus it is possible to accept on both the classical and the quantum tests, but still have $\rho$ be far from $\sigma$. We give an example here:
\begin{example}
    Assume we have the unknown state $\rho$, and after passing it through the quantum instrument, we transmit the state
    \begin{equation*}
        \rho_{CQ}=\frac{1}{2}|0\rangle\langle0|_Q\otimes|0\rangle\langle0|_C+\frac{1}{2}|1\rangle\langle1|_Q\otimes|1\rangle\langle1|_C,
    \end{equation*}
    in other words, it is a state where the quantum qubit equals the classical bit. Assume we have a known state $\sigma$, and under the same quantum instrument, we produce the state
    \begin{equation*}
        \sigma_{CQ}=\frac{1}{2}|1\rangle\langle1|_Q\otimes|0\rangle\langle0|_C+\frac{1}{2}|0\rangle\langle0|_Q\otimes|1\rangle\langle1|_C.
    \end{equation*}
    Now, one can try the above, analyzing the marginals obtained by classical communication and the quantum communication separately. In the classical communication, we have that
    \begin{align*}
        \rho_C&=\frac{1}{2}|0\rangle\langle0|+\frac{1}{2}|1\rangle\langle1|\\
        \sigma_C&=\frac{1}{2}|0\rangle\langle0|+\frac{1}{2}|1\rangle\langle1|.
    \end{align*}
    In other words, both see each bit with probability half. A classical only test would see no difference in these two marginals, and accept. Computing the quantum marginals, one would also see that
    \begin{align*}
        \rho_Q&=\frac{1}{2}|0\rangle\langle0|+\frac{1}{2}|1\rangle\langle1|\\
        \sigma_Q&=\frac{1}{2}|1\rangle\langle1|+\frac{1}{2}|0\rangle\langle0|.
    \end{align*}
    These two marginals are also equal, so the quantum only test would accept this part. So both have accepted, yet the corresponding outputs of the quantum instrument are actually far in trace distance.
\end{example}
It is clear from this example that if one wants to succeed in the mixed-communication setting, one must have an algorithm that looks at the classical bits along with the quantum bits together, and not one that isolates each part. Since the remainder of this section is dense, we give a brief overview of the steps we take to prove this upper bound.
\subsection{Outline of Proof}\label{sec:outline}
Overall, our approach to proving a matching upper bound can be broken into four main steps:
\begin{enumerate}
    \item Specific analysis of the BOW estimator for ``smoothed'' quantum instruments. At a high level, we aim to show that the BOW estimator actually has variance that depends on $K$, assuming the CQ states have probability mass that is roughly $O(1/K)$ over each of the classical labels.
    \item Specific design of quantum instruments which, for a fixed state, have a high probability to spread probability mass over all the classical labels, such that each has probability mass $O(1/K)$.
    \item Proof that these quantum instruments satisfy a compression property similar to \cite[Lemma 3.3]{doosti2026distributedquantumpropertytesting}. 
    \item Explicit design of an algorithm using the above.
\end{enumerate}
We now discuss each of these four points in detail, giving the rationale, and the technical details behind each.

\paragraph{Specific analysis of the BOW estimator.} Since our goal is to detect large trace distance, a natural idea is to use the same approach as in the quantum-only communication algorithm of Doosti \etal: that is, use the BOW tester (which relies on the BOW estimator as a subroutine) on the CQ states to test trace distance. After all, the BOW tester works on all quantum states, and CQ states are just specific quantum states. The issue with this is that the variance in the BOW estimator does not have any dependence on $K$, and thus such an approach will not, at least na\"ively, benefit from having classical bits. Moreover, improving the variance bound in the BOW estimator to obtain a dependence on $K$ also seems like a difficult task. An insight, however, is that the very generality of the BOW estimator is where we can try to make improvements, as we only need a good estimator for CQ states, not arbitrary ones. The idea then is to go over the analysis for the BOW estimator, and get a variance that depends, at a high level, on how much the probability spreads over each of the classical labels. Importantly, this does \emph{not} lead to a generic improvement to the BOW tester: indeed, if the CQ state concentrates on one (or, more broadly, few) of the classical strings, with no mass on many, we recover the variance of the BOW tester. Yet, whenever the CQ states spread mass ``roughly uniformly'' across all the classical labels (each label appears with probability $O(1/K)$), we yields a sharper bound on the variance, which crucially depends on $K$, and thus allows us, when applying Chebyshev's inequality, to get a dependence on $K$.

\paragraph{Quantum Instrument Design.} The goal, given the above analysis of the BOW estimator, is then to design quantum instruments that, given the states $\rho$ and $\sigma$, sufficiently spread the communication over all (or most) of the classical labels. Roughly speaking, we want the distributed node to communicate using, on average, almost all of the classical strings available to it. Since the BOW estimator needs, to benefit from this improved variance, each classical label to appear with probability $O(1/K)$, we aim for quantum instruments which use each classical string with probability $O(1/K)$. To this end, we show that the instrument analogue of the quantum channel used by Doosti \etal~suffices: that is, a quantum instrument defined by (1) sampling a Haar random unitary matrix and conjugating a state with it, then (2) measuring the classical register to create the classical message, (3) traces out an unused register, and (4)~sends the remaining state as the quantum message.

\paragraph{Instrument Compression.} The next thing we require is that our instruments satisfy a similar compression result as the channels do for Doosti \etal. At a high level, we require that the instruments do not ``contract'' the Hilbert--Schmidt distance too much, so that we can still distinguish when $\rho=\sigma$ and when $\|\rho-\sigma\|_1\geq\eps$. The approach to this structurally is the same as used by Doosti \etal: one can think of our overall approach as adapting their method to quantum instruments, though we also give a few simplifications and alternative arguments along the way.

\paragraph{Algorithm Design} Having proven the above three, the algorithm we use is akin to the public-coin version of \cref{alg:private_coin_no_classical} with quantum instruments. That is, instead of channels, we use quantum instruments, and instead of choosing a list of good unitaries upfront, we sample them with public randomness, and argue they are good with high probability.

We will now go through each of the four steps. We will introduce the preliminaries and other necessary results within each section. 
\subsection{Analysis of the BOW Estimator}
We start by defining some notation. First, we assume $\rho$ and $\sigma$ are CQ states, meaning that they are of the form
\begin{equation*}
    \rho=\sum_{c=1}^K\rho_c\otimes|c\rangle\langle c|_C, \quad\sigma=\sum_{c=1}^K\sigma_c\otimes|c\rangle\langle c|_C.
\end{equation*}
When it is clear, we also drop the subscript $C$ for the classical register. We also let
\begin{equation*}
    r_c\eqdef \tr[\rho_c], \quad s_c\eqdef\tr[\sigma_c],
\end{equation*}
which are the classical marginal probabilities of seeing the label $c$ under $\rho$ and $\sigma$. Finally, we also let
\begin{equation*}
    r_\infty\eqdef \max_{c}r_c, \quad
    \chi_r\eqdef \sum_cr_c^2, \quad
    \chi_s=\sum_c s_c^2,
\end{equation*}
which are quantities which capturing how ``spread out'' $\rho$ and $\sigma$ are over the classical labels. With this in hand, the main result we prove is the following:
\begin{theorem}\label{thm:cq_hs_estimator}
    Given $m$ copies of a CQ state $\rho$, and full knowledge of a CQ state $\sigma$, there is an estimator $\hat{\Gamma}$ such that
    \begin{equation*}
        \mathbb{E}[\hat{\Gamma}]=\Gamma=D_{HS}^2(\rho,\sigma),
    \end{equation*}
    and 
    \begin{equation*}
        \Var(\hat\Gamma)\leq O\left(\frac{r_\infty\Gamma}{m}+\frac{\chi_r+\chi_s}{m^2}\right).
    \end{equation*}
\end{theorem}
(We believe this result could be made more general; however, the version above suffices for our purposes.)\footnote{For instance, the BOW estimator works for $m$ copies of $\rho$ and $n$ copies of $\sigma$, whereas we assume unlimited copies of $\sigma$. In fact, our design of the estimator also implicitly is built by having full knowledge of $\sigma$, whereas the BOW estimator does not make this assumption.} We also emphasize again that this is \emph{not} a general improvement over the results of Badescu \etal: for instance, if the CQ states concentrate on one classical label, then  $r_\infty=1$, and $\chi_r=\chi_s=1$, and we recover the result in \cref{lemma:hilbert_schmidt_estimator}. However, if all $\chi_r=\chi_s=O(1/K)$, which implies $r_\infty=O(1/\sqrt K)$, then we do get a genuine improvement on the variance. This latter regime is the one we care about.

To prove this result, we follow the analysis of Badescu \etal~to design the estimator and to bound the variance. Though we adapt the arguments to work for CQ states, much of the analysis is exactly the same as Badescu \etal, with careful modifications to obtain the more fine-grained variance bound. Due to the large similarity, we defer the proof to the \cref{app:BOW_analysis}.
\subsection{Quantum Instrument Design}
Recall, in light of the above variance results, we need a quantum instrument which sufficiently spreads out over the classical labels. In this section, we give such a design. Without loss of generality, let us assume $d=Kd_qb$. In other words, we define our input Hilbert space $\mathbb{C}^d\cong C\otimes Q\otimes B$, where $\dim C=K,~\dim Q=d_q, \dim B=b$. (As before, if $d$ is not divisible by $Kq$, one can pad the input space, with only a constant factor loss.) At a high level, $C$ will be the classical register we send, $Q$ will be the quantum register we send, and $B$ will be the remaining register we trace out. To do this, we first design a quantum instrument based on sampling Haar-random unitaries, which we call a \emph{Haar instrument}. This instrument, to the best of our knowledge, is novel, and may find interest in more general quantum communication.
\begin{definition}[Haar instrument]\label{def:Haar_instrument}
    A \emph{Haar instrument} is a quantum instrument constructed as follows. First, sample a $d$-dimensional Haar-random unitary $U$. Then for each classical label $c\in[K]$, define the completely positive branch
    \begin{equation*}
        \mathcal{I}_{U,c}(X)\eqdef \tr_B[(\langle c|_C\otimes I_{QB})UXU^\dagger(|c\rangle_C\otimes I_{QB}]=\tr_B(A_c(U)).
    \end{equation*}
    Then, the full quantum instrument is the completely positive trace preserving map
    \begin{equation*}
        \mathcal{I}_U\eqdef\sum_{c=1}^K\mathcal{I}_{U,c}(X)\otimes|c\rangle\langle c|_C.
    \end{equation*}
\end{definition}
Roughly speaking, the user of this instrument conjugates $U$ with their state $X$, traces out the $B$ register, then measures $C$ in the computational basis to obtain a string $c$, which they send along with the $Q$ register. We now argue that this is a valid instrument.
\begin{lemma}\label{lemma:haar_instrument_is_instrument}
    For every positive semidefinite $X\succeq0$, each branch $\mathcal{I}_{U,c}(X)$ is also positive semidefinite. Moreover, $\mathcal{I}_U$ preserves traces, and if $\rho$ is a density matrix, then $\mathcal{I}_U(\rho)$ is a normalized CQ state.
\end{lemma}
\begin{proof}
    To see the first claim, notice each branch consists of a unitary conjugation, projection to a subspace corresponding to a classical label $c$, and then a partial trace over $B$. Each of these operations is completely positive, that is, if $X\succeq0$ then $\mathcal{I}_{U,c}(X)\succeq0$. Let us define the projector
    \begin{equation*}
        \Pi_c\eqdef |c\rangle\langle c|_C\otimes I_Q\otimes I_B.
    \end{equation*}
    The collection of all $K$ of these projectors are orthogonal, and sum to the identity. Then, we have that $\tr[\mathcal{I}_{U,c}(X)]=\tr[\Pi_cUXU^\dagger]$. Thus,
    \begin{align*}
        \tr[\mathcal{I}_U(X)]&=\sum_{c=1}^K\tr[\mathcal{I}_{U,c}(X)]
        =\tr\left[\left(\sum_{c=1}^K\Pi_c\right)UXU^\dagger\right]
        =\tr[UXU^\dagger]
        =\tr[X].
    \end{align*}
    Hence, the instrument is trace-preserving, so if $\rho$ is a density matrix, the output $\mathcal{I}_U(\rho)$ is a normalized CQ state.
\end{proof}
We now move to show this instrument does sufficiently spread out across the classical labels. Since we will work with pairs of states $\rho$, and hence pairs of instruments, we will use the following result, mostly due to Mele~\cite{DBLP:journals/quantum/Mele24} (our contribution being to refine the statement by providing explicit values for the coefficients).
\begin{lemma}\label{lemma:Haar_pair}
    Let $\rho$ be a density matrix on $\mathbb{C}^d$, and define $\eta\eqdef\tr[\rho^2]$. Then, $M=\mathbb{E}_U[(U\rho U^\dagger)^{\otimes 2}]=a_1I+a_2F$, where $F$ is the swap operator, and 
    \begin{equation*}
        a_1=\frac{d-\eta}{d(d^2-1)},~a_2=\frac{d\eta-1}{d(d^2-1)}.
    \end{equation*}
\end{lemma}
\begin{proof}
    We outline the main contributions by Mele, which obtains the general form of the expectation. Our contribution is to work out the exact value of the coefficients $a_1,a_2$. First, this expectation is what is referred to by Mele as a $k$-th moment operator, with $k=2$ \cite[Definition 4]{DBLP:journals/quantum/Mele24}. They then show that for $k=2$, this moment operator is spanned by the identity and the swap operator \cite[Theorem 7, Theorem 9, Example 14]{DBLP:journals/quantum/Mele24}. This is enough to give us that $\mathbb{E}_U[(U\rho U^\dagger)^{\otimes 2}]=a_1I+a_2F$. To work out the constants, we use properties of $\rho$ and equate traces. First, note that $\tr[M]=\tr(\rho\otimes\rho)=1$, since unitaries do not alter the trace, and it is just an expectation over density matrices over the product space. On this joint space $\mathbb{C}^d\otimes\mathbb{C}^d$, it is known that $\tr[I]=d^2$, and $\tr[F]=d$. Hence, we have our first restraint:
    \begin{equation*}
        a_1d^2+a_2d=1.
    \end{equation*}
    Next, we can apply a swap operator to $M$, and then take the trace. On one hand, this gives us
    \begin{align*}
        \tr[F(a_1I+a_2F)]&=a_1\tr[F]+a_2\tr[F^2]
        =a_1\tr[F]+a_2\tr[I]
        =a_1d+a_2d^2.
    \end{align*}
    On the other hand, using the fact that $\tr[F(A\otimes B)]=\tr[AB]$, we have that
    \begin{align*}
        \tr[F(U\rho U^\dagger\otimes U\rho U^\dagger)]&=\tr[(U\rho U^\dagger)^2]
        =\tr[\rho^2]
        =\eta.
    \end{align*}
    Hence, we get our second constraint
    \begin{equation*}
        a_1d+a_2d^2=\eta.
    \end{equation*}
    Solving for $a_1$ and $a_2$ gives the constants above.
\end{proof}
Next, we can give a result related to the probability of a particular classical label $c$.
\begin{lemma}\label{lemma:single_string_second_moment}
    Let $\rho$ be a density matrix, and let $\rho_c(U)=\mathcal{I}_{U,c}(\rho)$. Recall that the probability to observe any specific classical outcome $c$ is $p^\rho_c(U)\eqdef \tr[\rho_c(U)]$. Then, for any $c$, we have
    \begin{equation*}
        \mathbb{E}_U\left[(p_c^\rho(U))^2\right]=a_1r^2+a_2r,
    \end{equation*}
    where $a_1$ and $a_2$ are the coefficients in \cref{lemma:Haar_pair}, and $r$ is the rank of the projector $\Pi_c=d/K$.
\end{lemma}
\begin{proof}
    By definition, we have that $p_c^\rho(U)=\tr[\Pi_cU\rho U^\dagger]$, the trace of $\rho$, after being conjugated by $U$ and projected onto the relevant string $c$. Hence, we have that
    \begin{equation*}
        (p_c^\rho(U))^2=\tr[(\Pi_c\otimes\Pi_c)(U\rho U^\dagger)^{\otimes 2}].
    \end{equation*}
    Taking expectations, and using \cref{lemma:Haar_pair}, this gives that
    \begin{equation}
        \mathbb{E}_U\left[(p_c^\rho(U))^2\right]=\tr[(\Pi_c\otimes\Pi_c)(aI+bF)].
    \end{equation}
    We can compute each of these terms. Notice $\tr[(\Pi_c\otimes\Pi_c)I]=\tr[\Pi_c]^2=r^2$. Next, $\tr[(\Pi_c\otimes\Pi_c)F]=\tr[\Pi_c^2]=\tr[\Pi_c]=r$. Putting these two together with the coefficients proves the result.
\end{proof}
With these two lemmas in hand, we can now prove our Haar instrument sufficiently spreads out the mass among the classical labels. For convenience, let us define:
\begin{equation*}
    \chi_\rho(U) \eqdef \sum_c(p_c^\rho(U))^2.
\end{equation*}
Showing that we have mass sufficiently spread across all possible labels is then equivalent to showing $\chi_\rho$ is small. We prove the following result:
\begin{lemma}\label{thm:Haar_spread}
    For every fixed density matrix $\rho$, we have that
    \begin{equation*}
        \mathbb{E}_U[\chi_\rho(U)]\leq \frac{2}{K}.
    \end{equation*}
\end{lemma}
\begin{proof}
    We have that, by using \cref{lemma:single_string_second_moment},
    \begin{equation*}
         \mathbb{E}_U[\chi_\rho(U)]=\sum_{c=1}^K\mathbb{E}_U\left[(p_c^\rho(U))^2\right]=L(a_1r^2+a_2r).
    \end{equation*}
    Plugging in the constants, and simplifying, one gets that this quantity simplifies to $\frac{d/K+1}{d+1}\leq \frac{2}{K}$.
\end{proof}
By Markov's inequality, we get the simple corollary:
\begin{corollary}\label{cor:Haar_spread}
    Let $\rho$ be any fixed density matrix. Then, for every $A>0$, we have that
    \begin{equation*}
        \Pr_U\left[\chi_\rho(U)>\frac{A}{K}\right]\leq \frac{2}{A}.
    \end{equation*}
\end{corollary}
Then, taking $A\geq16$, we can use a union bound to show that with good probability, the chosen Haar instrument will spread out both $\rho$ and $\sigma$.
\begin{corollary}\label{cor:Haar_double_good_spread}
    Let $\rho$ and $\sigma$ be fixed density matrices. Then, with probability at least $3/4$ over the Haar-random U sampled with public-coin randomness, we have $\chi_\rho(U), \chi_\sigma(U)\leq 16/K$. Consequently, with probability at least $3/4$, $\chi_\rho(U)+\chi_\sigma(U)\leq 32/K$.
\end{corollary}
In particular, with high constant probability our Haar instrument does not concentrate too much on any one classical label, and spreads its probability mass around~--~achieving the second item of our overall blueprint for a public-coin algorithm. We can also straightforwardly bound the maximum mass $p_\infty$ as follows:
\begin{corollary}\label{cor:maximum_mass_bound}
    Let $\rho$ be a fixed unitary and $p_\infty^\rho\eqdef \max_cp_c^\rho(U)$. Then, with probability at least $7/8$, $p^\rho_\infty(U)\leq 4/\sqrt K$.
\end{corollary}
\begin{proof}
    By \cref{cor:Haar_spread}, with probability $7/8$, we have that $\chi_\rho(U)\leq 16/K$. Then, since $(p_\infty^\rho)^2\leq \chi_\rho(U)$, we get $p^\rho_\infty(U)\leq 4/\sqrt K$.
\end{proof}
\begin{remark}
    One might ask, given that we now have a CQ HS estimator with sufficiently good variance, can we also try to adapt this to get a private-coin bound. After all, in the private-coin bound without classical communication, we used the HS estimator as a subroutine. The answer here is ``not directly:'' the key reason is that all of our results only hold for a \emph{fixed} density matrix. This is fine in the public-coin protocol, since the input states $\rho$ and $\sigma$ are first fixed, and then the nodes sample the unitaries using public randomness. However, in the private-coin model, the unitaries are first fixed (and thus so are the Haar instruments), and for each such choice there exists $\rho$ such that the instrument is completely supported on one, or very few, classical labels. Then, one cannot obtain any better variance bound than that of Badescu \etal~from our refined analysis with CQ states. In short, we have a high probability statement over $U$ for each input state, while for private-coin protocols one would need a statement that holds uniformly for all states, for fixed $U$.
\end{remark}
\subsection{Quantum Instrument Compression}
With the Haar instruments defined, we aim now to show the second desired property:  that they do not compress the HS distance too much. To do so, we will again adapt the argument by Doosti \etal~for quantum channels. Write as before $d=Kd_qb$, and furthermore let $D=Kd_q$, so $\dim(C\otimes Q)=D$. For ease, , let $\Upsilon_U(X)\eqdef \|\mathcal{I}_U(X)\|_2^2$, for any traceless Hermitian $X$. Due to the orthogonal classical labels, this implies $\Upsilon_U(X)=\sum_{c=1}^K\|\mathcal{I}_{U,c}(X)\|_2^2$. We will then bound the second and fourth moments under the behavior of the Haar instrument, then use the Paley-Zygmund inequality (\cref{lemma:Paley_Zygmund}) to get a compression bound. We can now first come up with the analogues of the swap representations of the channels. The main result we will prove in this subsection is the following:
\begin{theorem}\label{thm:Haar_compression}
    For every traceless Hermitian $X$, there is an absolute constant $\gamma>0$ such that
    \begin{equation*}
        \Pr_U\left[\|\mathcal{I}_U(X)\|_2^2\geq \frac{d_q}{4d}\|X\|_2^2\right]\geq \gamma.
    \end{equation*}
    Moreover, one can take $\gamma=\frac{1}{3456C_W}$, where $C_W$ is a Weingarten coefficient.\footnote{One can find the exact constant by computing the Weingarten function, explicitly given in \cite{collins2006integration}, but for $S_4$, one can think of this as 20.}
\end{theorem}
 The proof of this theorem relies on the two lemmas below, whose proof can be found in~\cref{app:moment_bounds}.
While we follow the same analysis of Doosti \etal~with Weingarten coefficients to prove \cref{lemma:upsilon_second_moment}, we provide an alternative proof of \cref{lemma:first_moment_upsilon}, which does not rely on Weingarten coefficients, but instead only uses simple linear algebra. 
 
\begin{lemma}\label{lemma:first_moment_upsilon}
    Let $X$ be a traceless Hermitian. Then,
    \begin{equation*}
        \mathbb{E}_U[\Upsilon_U(X)]=\frac{d_qd-b}{d^2-1}\|X\|_2^2.
    \end{equation*}
    In particular, if $Kd_q^2\geq 2$, then
    \begin{equation*}
        \mathbb{E}_U[\Upsilon_U(X)]\geq \frac{1}{2}\cdot\frac{d_q}{d}\|X\|_2^2.
    \end{equation*}
\end{lemma}
\begin{lemma}\label{lemma:upsilon_second_moment}
    There is an absolute constant C, such that for every traceless Hermitian $X$, 
    \begin{equation*}
        \mathbb{E}_U[\Upsilon_U(X)^2]\leq C\left(\frac{d_q}{d}\right)^2\|X\|_2^4.
    \end{equation*}
    That is, $\mathbb{E}_U[\Upsilon_U(X)^2]\leq O\left(\left(\frac{d_q}{d}\right)^2\|X\|_2^4\right)$.
\end{lemma}
Assuming the above two lemmas, we show how to obtain \cref{thm:Haar_compression}.
\begin{proof}[Proof of~\cref{thm:Haar_compression}]
    We take $Y=\Upsilon_U(X)$. By \cref{lemma:first_moment_upsilon}, $\mathbb{E}[Y]\geq \frac{d_q}{2d}\|X\|_2^2$. Then, by \cref{lemma:upsilon_second_moment}, \smash{$\mathbb{E}[Y^2]\leq C\cdot \frac{d_q^2}{d^2}\|X\|_2^4$}. Applying the Paley-Zygmund inequality (\cref{lemma:Paley_Zygmund}) with $\theta=1/2$, we get
    \begin{equation*}
        \Pr[Y\geq \frac{1}{2}\mathbb{E}[Y]]\geq \frac{1}{4}\cdot\frac{1}{4C}=\frac{1}{16C}.
    \end{equation*}
    Then, on this event, we have that
    \begin{equation*}
        Y\geq\frac{1}{2}\mathbb{E}[Y]\geq \frac{d_q}{4d}\|X\|_2^2, 
    \end{equation*}
    as required.
\end{proof}
In particular for our quantum instruments, this implies the following by standard norm bounds.
\begin{corollary}\label{cor:post_instrument_gap}
    Let $\rho$ and $\sigma$ be density matrices such that $\|\rho-\sigma\|_1\geq\eps$. Then with constant probability over the choice of a Haar-random $U$, we have that
    \begin{equation*}
        \|\mathcal{I}_U(\rho-\sigma)\|_2^2\geq \frac{d_q\eps^2}{d^2}.
    \end{equation*}
\end{corollary}
\subsection{A Public-Coin Algorithm}
With all that setup complete, we can finally give our algorithm in the public-coin setting, with both classical and quantum communication, shown in \cref{alg:public_coin}.
\begin{algorithm}[H]
    \caption{Public-coin distributed algorithm with both classical and quantum communication}\label{alg:public_coin}
    \begin{algorithmic}[1]
        \State Set $L\eqdef O(1)$, $\beta=\eqdef \frac{1}{4}\frac{d_q\eps^2}{d^2}$, and $t\eqdef O\left(\frac{1}{\sqrt{K}\beta}\right)$. Set $m\eqdef Lt$.
        \State\Comment{All nodes agree beforehand on a numbering $N_{i,j}$ for $i\in[t]$ and $j\in[L]$.}
        \State Using public randomness, all nodes jointly sample $L$ different Haar random unitary matrices, $U_1,...,U_L$, and create the corresponding quantum instruments $\mathcal{I}_{U_i}$.
        \State Each distributed node applies $\rho$ to the quantum instrument $\mathcal{I}_{U_i}$, to send a classical label $c$, with $c\in[K]$ consisting of $n_c$ bits, and sends the $Q$ register, consisting of a $d_q$-dimensional qudit.
        \State The referee node, with full knowledge of $\sigma$ feeds $\sigma$ through the quantum instrument to get $\tau_{i}=\mathcal{I}_{U_i}(\sigma)$, for each sampled unitary. 
        \State\label{algo:publiccoin:checkspreadout} The referee node checks the classical labels are sufficiently spread out by computing $\chi_{\tau_i}<\hat C/K$, where $C>0$ is an absolute constant. If not, it records $\reject$ for the corresponding instrument $i$, and only proceeds with what is below for the $L' \leq L$ instruments sufficiently spread out.
        \For {$i=1$ to $L'$}
            \State Take all messages corresponding to the $i$-th quantum instrument, which gives a CQ state $\rho_i$. Use the CQ HS estimator obtain an estimate $\hat{\Gamma}_i$ for $\Gamma_i=\|\tau_i-\rho_i\|_2^2$.
            \If {$\hat\Gamma> \beta/2$}
                \State record \reject,
            \Else \State record \accept.
            \EndIf
        \EndFor
        \State\Return \accept{} if the majority of the loops accept, else \reject.
    \end{algorithmic}
\end{algorithm}
In Line~\ref{algo:publiccoin:checkspreadout}, the referee node has to verify whether the quantum instruments sufficiently spread out $\sigma$. For analysis sake, we assume that all sampled unitaries produce good instruments, meaning it sufficiently spreads out over all classical labels for both $\rho$ and $\sigma$. By \cref{cor:Haar_double_good_spread}, we know this occurs with good probability, and hence we can always just adjust the (constant) number of repetitions $L$, costing only a constant factor in the result. We will show the following:
\begin{theorem}\label{thm:public_coin_ub}
    Assume the distributed nodes are allowed to communicate a $d_q$-dimensional qudit, and an $n_c$-length classical string. Then, to $\eps$-certify a $d$-dimensional state with probability at least $2/3$, it suffices to take $m= O\left(\frac{d^2}{d_q^2\sqrt{K}\eps^2}\right)$, i.e.,
    \begin{equation*}
        m= O\left(\frac{d^2}{2^{n_q+n_c/2}\eps^2}\right).
    \end{equation*}
\end{theorem}
\begin{proof}
    Let $X=\rho-\sigma$. Then, by \cref{thm:Haar_compression}, with constant probability over a public-coin choice of unitary $U$, we have that
    \begin{equation*}
        \|\mathcal{I}_U(X)\|_2^2\geq \frac{d_q}{4d}\|X\|_2^2.
    \end{equation*}
    If $\rho=\sigma$, this quantity equals zero, while if $\|\rho-\sigma\|_1 \geq \eps$ this is at least $\frac{1}{4}\cdot\frac{d_q\eps^2}{d}$ by~ \cref{cor:post_instrument_gap}. Recall that the algorithm set a threshold at $\beta/2$, where $\beta\eqdef\frac{1}{4}\cdot\frac{d_q\eps^2}{d}$. Next, by \cref{cor:Haar_double_good_spread} we know that with good probability, the quantum instrument $I_{U_i}$ outputs CQ states $\rho_i$ and $\tau_i$ that are spread over all classical strings. In other words, we have that 
    \begin{equation*}
        \chi_{\rho_i}, \chi_{\tau_i}\leq O(1/K).
    \end{equation*}
    Conditional on this happening, the CQ HS estimator then has variance, by \cref{thm:cq_hs_estimator}
    \begin{equation*}
        \Var[\hat\Gamma_i]\leq O\left(\frac{\Gamma_i}{\sqrt{K}t}+\frac{1}{Kt^2}\right).
    \end{equation*}
    Now we consider two separate cases. First, assume $\rho=\sigma$, so $\Var[\hat\Gamma_i]\leq O\left(\frac{1}{Kt^2}\right)$. Then, by Chebyshev's inequality, we have that
    \begin{equation*}
        \Pr[\hat\Gamma_i\geq \beta/2]\leq O\left(\frac{\Var[\hat\Gamma_i]}{\beta^2}\right)=O\left(\frac{1}{Kt^2\beta^2}\right),
    \end{equation*}
    which is at most $1/3$ by setting the constant in $t$ large enough. Considering the other case, if $\|\rho-\sigma\|_1\geq\eps$, then $\Gamma_i\geq\beta$, so $\Gamma_i-\beta/2\geq \Gamma_i/2$. Then applying Chebyshev's in this case, we get that
    \begin{align*}
        \Pr[\hat\Gamma_i<\beta/2]&\leq\Pr[|\hat\Gamma_i-\Gamma_i|\geq \Gamma_i/2]\\
        &\leq O\left(\frac{\Var[\hat\Gamma_i]}{\Gamma_i^2}\right)\\
        &=O\left(\frac{1}{\sqrt{K}t\Gamma_i}+\frac{1}{Kt^2\Gamma_i^2}\right)\\
        &\leq O\left(\frac{1}{\sqrt{K}t\beta}+\frac{1}{Kt^2\beta^2}\right),
    \end{align*}
    which again is at most $1/3$ by increasing the constant in $t$. Then, the total number of nodes used here is
    \begin{equation*}
        m=Lt=O\left(\frac{d^2}{d_q^2\sqrt{K}\eps^2}\right),
    \end{equation*}
    establishing the theorem.
\end{proof}
\section{Conclusion}
In this work, we have shown lower bounds for distributed quantum state certification in both the public-coin and private-coin models, where classical and quantum communication are restrained but permitted. We have also shown that these lower bounds are tight in the public-coin model. We leave several avenues for followup work.
\paragraph{Tight Private-Coin Algorithm Without Classical Communication.} The algorithm we exhibit for the private-coin model when classical communication is not permitted in \cref{thm:private_coin_no_classical_ub} is tight up to a logarithmic factor. We leave it open to see whether this can be improved: note that this logarithmic factor is due to the use of the Matrix Chernoff bound to show that a list of good unitaries (\cref{thm:good_unitaries}) to argue the existence of a ``good'' set of unitaries. We believe our analysis of~\cref{thm:good_unitaries} to be tight in that respect, and so it seems a new idea would be needed.
\paragraph{Private-Coin Algorithms in the General Setting.} We also leave open the broader question of designing a private-coin algorithm when both classical and quantum communication are allowed. As we remarked above, the reason for why our public-coin methods do not generalize to the private-coin setting is that they only hold for fixed states. In other words, we can find good quantum instruments for fixed states, but the private-coin would require us to find quantum instruments that are good uniformly over all states. We believe a new approach is needed to generalize the result to the mixed communication setting: for instance, one could try to adapt our private-coin algorithm with no classical communication, to prove a list of good unitaries exist, and combine it with our refined analysis of the BOW estimator as a starting point.
\paragraph{Unequal Samples, Messages and Mixed Coins.} Recently, Canonne \etal~extended the problem of distributed, communication constrained Gaussian mean estimation in the classical setting by considering generalizations where the distributed nodes share a small number of random bits, rather than sharing all their randomness, where each distributed node holds a different number of samples, and where each user can send a different number of bits to the central node \cite{canonne2026distributedgaussianmeantesting}. These generalizations could apply here as well, and it would be interesting to see what results are obtained under these generalizations.
\paragraph{Acknowledgments.} The author would like to thank Ryan Sweke and Cl\'{e}ment Canonne for helpful discussions. The author would additionally like to thank Cl\'{e}ment Canonne for comments and assistance on the draft.
\printbibliography

@misc{doosti2026distributedquantumpropertytesting,
      title={Distributed Quantum Property Testing with Communication Constraints}, 
      author={Mina Doosti and Ryan Sweke and Chirag Wadhwa},
      year={2026},
      eprint={2604.05962},
      archivePrefix={arXiv},
      primaryClass={quant-ph},
      url={https://arxiv.org/abs/2604.05962}, 
}

@article{TemmeDivergence,
    author = {Temme, K. and Kastoryano, M. J. and Ruskai, M. B. and Wolf, M. M. and Verstraete, F.},
    title = {The  {$\chi^2$}-divergence and mixing times of quantum Markov processes},
    journal = {Journal of Mathematical Physics},
    volume = {51},
    number = {12},
    pages = {122201},
    year = {2010},
    issn = {0022-2488},
    doi = {10.1063/1.3511335},
    url = {https://doi.org/10.1063/1.3511335},
    eprint = {https://pubs.aip.org/aip/jmp/article-pdf/doi/10.1063/1.3511335/13353580/122201_1_online.pdf},
}

@misc{odonnell2025instanceoptimalquantumstatecertification,
      title={Instance-Optimal Quantum State Certification with Entangled Measurements}, 
      author={Ryan O'Donnell and Chirag Wadhwa},
      year={2025},
      eprint={2507.06010},
      archivePrefix={arXiv},
      primaryClass={quant-ph},
      url={https://arxiv.org/abs/2507.06010}, 
}

@book{Wilde_2016,
   title={Quantum Information Theory},
   ISBN={9781316809976},
   url={http://dx.doi.org/10.1017/9781316809976},
   DOI={10.1017/9781316809976},
   publisher={Cambridge University Press},
   author={Wilde, Mark M.},
   year={2016},
   month=Nov }

@misc{hashim2026understandingquantuminstruments,
      title={Understanding Quantum Instruments}, 
      author={Akel Hashim},
      year={2026},
      eprint={2604.18884},
      archivePrefix={arXiv},
      primaryClass={quant-ph},
      url={https://arxiv.org/abs/2604.18884}, 
}

@inproceedings{DBLP:conf/stoc/BadescuO019,
  author       = {Costin Badescu and
                  Ryan O'Donnell and
                  John Wright},
  editor       = {Moses Charikar and
                  Edith Cohen},
  title        = {Quantum state certification},
  booktitle    = {Proceedings of the 51st Annual {ACM} {SIGACT} Symposium on Theory
                  of Computing, {STOC} 2019, Phoenix, AZ, USA, June 23-26, 2019},
  pages        = {503--514},
  publisher    = {{ACM}},
  year         = {2019},
  url          = {https://doi.org/10.1145/3313276.3316344},
  doi          = {10.1145/3313276.3316344},
  bibsource    = {dblp computer science bibliography, https://dblp.org}
}

@article{Tropp_2011,
   title={User-Friendly Tail Bounds for Sums of Random Matrices},
   volume={12},
   ISSN={1615-3383},
   url={http://dx.doi.org/10.1007/s10208-011-9099-z},
   DOI={10.1007/s10208-011-9099-z},
   number={4},
   journal={Foundations of Computational Mathematics},
   publisher={Springer Science and Business Media LLC},
   author={Tropp, Joel A.},
   year={2011},
   month=Aug, pages={389–434} }

@article{DBLP:journals/cn/CaleffiAFIMC24,
  author       = {Marcello Caleffi and
                  Michele Amoretti and
                  Davide Ferrari and
                  Jessica Illiano and
                  Antonio Manzalini and
                  Angela Sara Cacciapuoti},
  title        = {Distributed quantum computing: {A} survey},
  journal      = {Comput. Networks},
  volume       = {254},
  pages        = {110672},
  year         = {2024},
  url          = {https://doi.org/10.1016/j.comnet.2024.110672},
  doi          = {10.1016/J.COMNET.2024.110672},
  bibsource    = {dblp computer science bibliography, https://dblp.org}
}

@article{DBLP:journals/tit/AcharyaCLST22,
  author       = {Jayadev Acharya and
                  Cl{\'{e}}ment L. Canonne and
                  Yuhan Liu and
                  Ziteng Sun and
                  Himanshu Tyagi},
  title        = {Interactive Inference Under Information Constraints},
  journal      = {{IEEE} Trans. Inf. Theory},
  volume       = {68},
  number       = {1},
  pages        = {502--516},
  year         = {2022},
  url          = {https://doi.org/10.1109/TIT.2021.3123905},
  doi          = {10.1109/TIT.2021.3123905},
  bibsource    = {dblp computer science bibliography, https://dblp.org}
}

@article{DBLP:journals/corr/abs-2401-09650,
  author       = {Yuhan Liu and
                  Jayadev Acharya},
  title        = {The role of shared randomness in quantum state certification with
                  unentangled measurements},
  journal      = {CoRR},
  volume       = {abs/2401.09650},
  year         = {2024},
  url          = {https://doi.org/10.48550/arXiv.2401.09650},
  doi          = {10.48550/ARXIV.2401.09650},
  eprinttype   = {arXiv},
  eprint       = {2401.09650},
  bibsource    = {dblp computer science bibliography, https://dblp.org}
}

@misc{AlAssal2014InvitationLieAlgebras,
  author       = {Al Assal, Fernando},
  title        = {Invitation to Lie Algebras and Representations},
  year         = {2014},
  note         = {University of Chicago REU paper},
  url          = {https://math.uchicago.edu/~may/REU2014/REUPapers/AlAssal.pdf},
}

@misc{Morgan2022SimpleLieAlgebrasCompactLieGroups,
  author       = {Morgan, John W.},
  title        = {Simple Lie Algebras and of Compact Lie Groups},
  year         = {2022},
  note         = {Lie Groups: Fall 2022, Lecture VI},
  url          = {https://www.math.columbia.edu/~jmorgan/LieGpLectureVI.pdf},
}

@article{HANSEN_2003,
   title={Jensen's operator inequality},
   volume={35},
   ISSN={1469-2120},
   url={http://dx.doi.org/10.1112/S0024609303002200},
   DOI={10.1112/s0024609303002200},
   number={04},
   journal={Bulletin of the London Mathematical Society},
   publisher={Wiley},
   author={Hansen, Frank and Pederson, Gert K.},
   year={2003}, 
   pages={553–564} 
}

@article{DBLP:journals/quantum/Mele24,
  author       = {Antonio Anna Mele},
  title        = {Introduction to Haar Measure Tools in Quantum Information: {A} Beginner's
                  Tutorial},
  journal      = {Quantum},
  volume       = {8},
  pages        = {1340},
  year         = {2024},
  url          = {https://doi.org/10.22331/q-2024-05-08-1340},
  doi          = {10.22331/Q-2024-05-08-1340},
  bibsource    = {dblp computer science bibliography, https://dblp.org}
}

@article{Collins_2010,
   title={Random Quantum Channels I: Graphical Calculus and the Bell State Phenomenon},
   volume={297},
   ISSN={1432-0916},
   url={http://dx.doi.org/10.1007/s00220-010-1012-0},
   DOI={10.1007/s00220-010-1012-0},
   number={2},
   journal={Communications in Mathematical Physics},
   publisher={Springer Science and Business Media LLC},
   author={Collins, Benoît and Nechita, Ion},
   year={2010},
   month=Feb, pages={345–370} 
}

@article{steele2004paley,
  title={The Paley-Zygmund argument and three variations},
  author={Steele, M},
  journal={Class note},
  year={2004}
}

@misc{etingof2011introductionrepresentationtheory,
      title={Introduction to representation theory}, 
      author={Pavel Etingof and Oleg Golberg and Sebastian Hensel and Tiankai Liu and Alex Schwendner and Dmitry Vaintrob and Elena Yudovina},
      year={2011},
      eprint={0901.0827},
      archivePrefix={arXiv},
      primaryClass={math.RT},
      url={https://arxiv.org/abs/0901.0827}, 
}

@article{collins2006integration,
  title={Integration with respect to the Haar measure on unitary, orthogonal and symplectic group},
  author={Collins, Beno\^it and {\'S}niady, Piotr},
  journal={Communications in Mathematical Physics},
  volume={264},
  number={3},
  pages={773--795},
  year={2006},
  publisher={Springer}
}

@inproceedings{DBLP:conf/stoc/ODonnellW15,
  author       = {Ryan O'Donnell and
                  John Wright},
  editor       = {Rocco A. Servedio and
                  Ronitt Rubinfeld},
  title        = {Quantum Spectrum Testing},
  booktitle    = {Proceedings of the Forty-Seventh Annual {ACM} on Symposium on Theory
                  of Computing, {STOC} 2015, Portland, OR, USA, June 14-17, 2015},
  pages        = {529--538},
  publisher    = {{ACM}},
  year         = {2015},
  url          = {https://doi.org/10.1145/2746539.2746582},
  doi          = {10.1145/2746539.2746582},
  bibsource    = {dblp computer science bibliography, https://dblp.org}
}

@InProceedings{pmlr-v247-liu24a,
  title = 	 {The role of randomness in quantum state certification with unentangled measurements},
  author =       {Liu, Yuhan and Acharya, Jayadev},
  booktitle = 	 {Proceedings of Thirty Seventh Conference on Learning Theory},
  pages = 	 {3523--3555},
  year = 	 {2024},
  editor = 	 {Agrawal, Shipra and Roth, Aaron},
  volume = 	 {247},
  series = 	 {Proceedings of Machine Learning Research},
  publisher =    {PMLR},
  url = 	 {https://proceedings.mlr.press/v247/liu24a.html},
}

@article{DBLP:journals/tit/AcharyaCT20,
  author       = {Jayadev Acharya and
                  Cl{\'{e}}ment L. Canonne and
                  Himanshu Tyagi},
  title        = {Inference Under Information Constraints {I:} Lower Bounds From Chi-Square
                  Contraction},
  journal      = {{IEEE} Trans. Inf. Theory},
  volume       = {66},
  number       = {12},
  pages        = {7835--7855},
  year         = {2020},
  url          = {https://doi.org/10.1109/TIT.2020.3028440},
  doi          = {10.1109/TIT.2020.3028440},
  bibsource    = {dblp computer science bibliography, https://dblp.org}
}

@book{355459,
author = {Peleg, David},
title = {Distributed computing: a locality-sensitive approach},
year = {2000},
isbn = {0898714648},
publisher = {Society for Industrial and Applied Mathematics},
address = {USA}
}

@article{DBLP:journals/corr/abs-2602-15529,
  author       = {Fabien Dufoulon and
                  Fr{\'{e}}d{\'{e}}ric Magniez and
                  Gopal Pandurangan},
  title        = {Tight Communication Bounds for Distributed Algorithms in the Quantum
                  Routing Model},
  journal      = {CoRR},
  volume       = {abs/2602.15529},
  year         = {2026},
  url          = {https://doi.org/10.48550/arXiv.2602.15529},
  doi          = {10.48550/ARXIV.2602.15529},
  eprinttype   = {arXiv},
  eprint       = {2602.15529},
  bibsource    = {dblp computer science bibliography, https://dblp.org}
}

@inproceedings{10.1145/3662158.3662767,
author = {Fraigniaud, Pierre and Luce, Ma\"{e}l and Magniez, Frederic and Todinca, Ioan},
title = {Even-Cycle Detection in the Randomized and Quantum CONGEST Model},
year = {2024},
isbn = {9798400706684},
publisher = {Association for Computing Machinery},
address = {New York, NY, USA},
url = {https://doi.org/10.1145/3662158.3662767},
doi = {10.1145/3662158.3662767},
booktitle = {Proceedings of the 43rd ACM Symposium on Principles of Distributed Computing},
pages = {209–219},
numpages = {11},
series = {PODC '24}
}

@inproceedings{DBLP:conf/focs/Chen0HL22,
  author       = {Sitan Chen and
                  Jerry Li and
                  Brice Huang and
                  Allen Liu},
  title        = {Tight Bounds for Quantum State Certification with Incoherent Measurements},
  booktitle    = {63rd {IEEE} Annual Symposium on Foundations of Computer Science, {FOCS}
                  2022, Denver, CO, USA, October 31 - November 3, 2022},
  pages        = {1205--1213},
  publisher    = {{IEEE}},
  year         = {2022},
  url          = {https://doi.org/10.1109/FOCS54457.2022.00118},
  doi          = {10.1109/FOCS54457.2022.00118},
  bibsource    = {dblp computer science bibliography, https://dblp.org}
}

@article{DBLP:journals/corr/abs-2408-17439,
  author       = {Yuhan Liu and
                  Jayadev Acharya},
  title        = {Quantum state testing with restricted measurements},
  journal      = {CoRR},
  volume       = {abs/2408.17439},
  year         = {2024},
  url          = {https://doi.org/10.48550/arXiv.2408.17439},
  doi          = {10.48550/ARXIV.2408.17439},
  eprinttype   = {arXiv},
  eprint       = {2408.17439},
  bibsource    = {dblp computer science bibliography, https://dblp.org}
}

@misc{canonne2026distributedgaussianmeantesting,
      title={Distributed Gaussian Mean Testing under Communication Constraints: messages, samples, and coins}, 
      author={Clément L. Canonne and Nimitt},
      year={2026},
      eprint={2605.29426},
      archivePrefix={arXiv},
      primaryClass={cs.DS},
      url={https://arxiv.org/abs/2605.29426}, 
}
\appendix
\section{Proof of \cref{lemma:instrument_minimax_maximin_bounds}}\label{app:lemma_instrument_minimax_maximin_bounds}
Here, we adapt the proof of Doosti \etal~to prove \cref{lemma:instrument_minimax_maximin_bounds}.
\begin{proof}
    In the public-coin setting, we will treat the action of distributed nodes as a set of quantum instruments $\{\mathcal{I}_i\}_{i\in[m]}$, where each $\mathcal{I}_i$ can depend on some shared random string $\mathbf{r}$. Let $\mathcal{A}(\rho,\mathbf{r})$ be the output of the referee node. Let $m$ be large enough to succeed with probability at least $3/4$. Then, we have that
    \begin{align*}
        \frac{1}{2}\Pr_{\rho\sim D}[\mathcal{A}(\rho,\mathbf{r})=\reject]+\frac{1}{2}\Pr_{\rho=I/d}[\mathcal{A}(\rho,\mathbf{r})=\accept]&\geq \frac{1}{2}\cdot\frac{3}{4}\cdot\Pr_{\rho\sim D}\left[\|\rho-I/d\|_1\geq\eps\right]+\frac{1}{2}\cdot\frac{3}{4}\\
        &\geq \frac{9}{16}.
    \end{align*}
    Then, for any $D$, there exists some random string $\mathbf{r}=r$ such that
    \begin{equation*}
        \frac{1}{2}\Pr_{\rho\sim D}[\mathcal{A}(\rho,\mathbf{r})=\reject|\mathbf{r}=r]+\frac{1}{2}\Pr_{\rho=I/d}[\mathcal{A}(\rho,\mathbf{r})=\accept|\mathbf{r}=r]\geq \frac{9}{16}.
    \end{equation*}
    Then, for each $D$, there is a deterministic choice of quantum instruments $\mathcal{I}_1^{(r)},...,\mathcal{I}_m^{(r)}$ such that the referee node can distinguish between the equally like cases of $\rho\sim D$ and $\rho=I/d$ with probability at least 9/16. Then, if we define
    \begin{equation*}
        P_r=\mathbb{E}_{\rho\sim D}\left[\bigotimes_{i=1}^m\mathcal{I}_i^{(r)}(\rho)\right],~Q_r=\bigotimes_{i=1}^m\mathcal{I}_i^{(r)}(I/d),
    \end{equation*}
    Although in this setting, the referee node obtains both a classical string and a quantum state, we can view the classical string as being stored in a qubit register, meaning we can still view the output of the quantum instrument as a quantum object with a density matrix. Then, we can still apply the Helstrom bound, meaning the optimal success probability for distinguishing equally likely $P_r$ and $Q_r$ is given by 
    \begin{equation*}
        \frac{1}{2}+\frac{1}{4}\|P_r-Q_r\|_1\geq \frac{9}{16}.
    \end{equation*}
    Using \cref{lemma:1_norm_qchi_inequality}, we get the bound for the public coin. 
    
    Moving to the private-coin model, each distributed node has its own internal randomness, $r_1,...,r_m$. Then, fixing the random strings each distributed node gets, this determines fixed quantum instruments $\mathcal{I}_1^{r_1},...,\mathcal{I}_m^{r_m}$. Then, we can take an expectation over all random strings, and model the action of each distributed node as a deterministic instrument $\mathcal{I}_i=\mathbb{E}_{r_i}[\mathcal{I}_i^{r_i}]$. This is a valid quantum instrument, since it is an average of quantum instruments, and hence is completely positive and trace preserving. Then, we can define for these instruments
    \begin{equation*}
        P=\mathbb{E}_{\rho\sim D}\left[\bigotimes_{i=1}^m\mathcal{I}_i(\rho)\right],~Q=\bigotimes_{i=1}^m\mathcal{I}_i(I/d).
    \end{equation*}
    By the same argument as in the public-coin case, we must have that for any $D$,
    \begin{equation*}
        \|P-Q\|_1\geq \frac{1}{4}.
    \end{equation*}
    Since this holds true for any $D$, it holds true for the minimum:
    \begin{equation*}
        \min_{D\in\mathcal{D}_\eps(I_d/d)}\|P-Q\|_1\geq \frac{1}{4}.
    \end{equation*}
    Finally, we maximize over all possible instruments, giving that
    \begin{equation*}
        \max_{\mathcal{I}_1,...,\mathcal{I}_m}\min_{D\in\mathcal{D}_\eps(I_d/d)}\QChi{\mathbb{E}_{\rho\sim D}\left[\bigotimes_{i=1}^m\mathcal{I}_i(\rho)\right]}{\bigotimes_{i=1}^m\mathcal{I}_i(I_d/d)}\geq\frac{1}{16}.
    \end{equation*}
    Using \cref{lemma:1_norm_qchi_inequality} again gives the result for private coins.
\end{proof}
\section{Symmetrized Quantum Ingster-Suslina Lemma}\label{app:symmetrized_quantum_ingster_suslina}
We prove in this section \cref{lemma:quantum_ingster-suslina_symm}, by following the work of O'Donnell \etal, but adapting the analysis for the different $\chi^2$-divergence \cite{odonnell2025instanceoptimalquantumstatecertification}.
\begin{proof}
    Let $R_\theta=\rho_\sigma^{(m)}$, $S=\sigma^{(m)}$, $P=\mathbb{E}_\theta[R_\theta]$ for convenience. Then, notice by the definition of the $\chi^2$ divergence, we have that
    \begin{align*}
        \QChi{\mathbb{E}_\theta[\rho^{(m)}_{\sigma}]}{\sigma^{(m)}}&=\QChi{P_\theta}{S}\\
        &=\tr[S^{-1/2}P_\theta S^{-1/2}P_{\theta'}]-1\\
        1+\QChi{P_\theta}{S}&=\tr[S^{-1/2}P_\theta S^{-1/2}P_{\theta'}]\\
        &=\tr[S^{-1/2}E_\theta[R_\theta] S^{-1/2}E_{\theta'}[R_{\theta'}]]\\
        &=\mathbb{E}_{\theta,\theta'}[\tr[S^{-1/2}R_\theta S^{-1/2}R_{\theta'}]]\\
        &=\mathbb{E}_{\theta,\theta'}\left[\tr\left[\bigotimes_{i=1}^m\sigma_i^{-1/2}\rho_{i,\theta}\sigma_{i}^{-1/2}\rho_{i,\theta'}\right]\right]\\
        &=\mathbb{E}_{\theta,\theta'}\left[\prod_{i=1}^m\tr[\sigma_i^{-1/2}\rho_{i,\theta}\sigma_{i}^{-1/2}\rho_{i,\theta'}]\right],
    \end{align*}
    where in line six, we revert to the full expressions and  use that the product of tensor products is the tensor product of products, and in line seven, we use that the trace of tensor products is the product of traces. Now, let $\Delta_{i,\theta}=\rho_{i,\theta}-\sigma_i$, so $\rho_{i,\theta}=\sigma_i+\Delta_{i,\theta}$. Then notice, we have that
    \begin{equation*}
        \tr[\Delta_{i,\theta}]=\tr[\rho_{i,\theta}-\sigma_i]=\tr[\rho_{i,\theta}]-\tr[\sigma_i]=0,
    \end{equation*}
    since both terms are quantum states. We can now rewrite the trace above in terms of $\Delta$ and $\sigma$:
    \begin{align*}
        \tr[\sigma_i^{-1/2}\rho_{i,\theta}\sigma_{i}^{-1/2}\rho_{i,\theta'}]&=\tr[\sigma_i^{-1/2}(\sigma_i+\Delta_{i,\theta})\sigma_i^{-1/2}(\sigma_i+\Delta_{i,\theta'})]\\
        &=\tr[\sigma_i^{-1/2}\sigma_i\sigma_i^{-1/2}\sigma_i]+\tr[\sigma_i^{-1/2}\sigma_i\sigma_i^{-1/2}\Delta_{i,\theta'}]\\
        &+\tr[\sigma_i^{-1/2}\Delta_{i,\theta}\sigma_i^{-1/2}\sigma_i]+\tr[\sigma_i^{-1/2}\Delta_{i,\theta}\sigma_i^{-1/2}\Delta_{i,\theta'}],
    \end{align*}
    where in the last line we expand and use linearity of traces. Notice now that we have, considering each of the four terms,
    \begin{align*}
        \tr[\sigma_i^{-1/2}\sigma_i\sigma_i^{-1/2}\sigma_i]=\tr[\sigma_i]&=1\\
        \tr[\sigma_i^{-1/2}\sigma_i\sigma_i^{-1/2}\Delta_{i,\theta'}]=\tr[\Delta_{i,\theta'}]&=0\\
        \tr[\sigma_i^{-1/2}\Delta_{i,\theta}\sigma_i^{-1/2}\sigma_i]=\tr[\Delta_{i,\theta}\sigma_i^{-1/2}\sigma_i\sigma_i^{-1/2}]=\tr[\Delta_{i,\theta}]&=0\\
        \tr[\sigma_i^{-1/2}\Delta_{i,\theta}\sigma_i^{-1/2}\Delta{i,\theta'}]&=Z_i(\theta,\theta').
    \end{align*}
    Hence, we have that
    \begin{align*}
        \tr[\sigma_i^{-1/2}\rho_{i,\theta}\sigma_{i}^{-1/2}\rho_{i,\theta'}]&=1+Z_i(\theta,\theta')\\
        1+\QChi{P_\theta}{S}&=\mathbb{E}_{\theta,\theta'}\left[\prod_{i=1}^m(1+Z_i(\theta,\theta'))\right].
    \end{align*}
    We can get the final inequality using that $1+x\leq e^x$ when $x\geq 0$.
\end{proof}
\section{BOW Estimator Analysis}\label{app:BOW_analysis}
We start by outlining the theorem we want to prove:
\begin{theorem}
    Given $m$ copies of a CQ state $\rho$, and full knowledge of a CQ state $\sigma$, there is an estimator $\hat{\Gamma}$ such that
    \begin{equation*}
        \mathbb{E}[\hat{\Gamma}]=\Gamma=D_{HS}^2(\rho,\sigma),
    \end{equation*}
    and 
    \begin{equation*}
        \Var(\hat\Gamma)\leq O\left(\frac{r_\infty\Gamma}{m}+\frac{\chi_r+\chi_s}{m^2}\right).
    \end{equation*}
\end{theorem}
Before going into the proof, it is worth overviewing the approach of Badescu \etal, since we will essentially mimic their arguments. To do this, we will need a few definitions and results on quantum probability, which we give now. These are all taken from \cite{DBLP:conf/stoc/BadescuO019}
\begin{definition}[Observable]
    An operator $X$, which is an endomorphism on $V$ (i.e., homomorphism from $V$ to $V$) is called an observable.
\end{definition}
\begin{definition}[Expectation, Covariance, Variance]
    Let $\rho$ be a quantum state on Hilbert space $V$. Then an observable $X$ has expectation
    \begin{equation*}
        \mathbb{E}_\rho[X]=\tr(\rho X).
    \end{equation*}
    For two such observables $X$ and $Y$, the covariance is defined as
    \begin{equation*}
        \Cov_\rho[X,Y]=\mathbb{E}_\rho[(X-\mu_XI)^\dagger(Y-\mu_YI)],
    \end{equation*}
    where $\mu_X=\mathbb{E}_\rho[X]$ and $\mu_Y=\mathbb{E}_\rho[Y]$. This also gives that $\Cov_\rho[X,Y]=\mathbb{E}_\rho[X^\dagger Y]-\mu_X^\dagger\mu_Y$. The variance is defined by
    \begin{equation*}
        \Var_\rho[X]=\Cov_\rho[X,X].
    \end{equation*}
    In particular, when $X$ is Hermitian, this recovers the recognizable result $\Var_\rho[X]=\mathbb{E}_\rho[X^2]=\mathbb{E}_\rho[X]^2$. When clear, we will also omit the subscripts. We will also use operator and observable interchangeably.
\end{definition}
We can then show the following properties.
\begin{lemma}\label{lemma:operator_expectation}
    Let $\rho$ and $\sigma$ be quantum states, with an operator $A$ acting on the first system, and $B$ on the second system. Then, $\mathbb{E}_{\rho\otimes\sigma}[A\otimes B]=\mathbb{E}_\rho[A]\mathbb{E}_\sigma[B]$.
\end{lemma}
\begin{proof}
    Just from the definition, we have
    \begin{equation*}
        \tr[(\rho\otimes\sigma)(A\otimes B)]=\tr[\rho A]\tr[\sigma B],
    \end{equation*}
    as required.
\end{proof}
This also gives the following covariance result:
\begin{corollary}\label{cor:operator_covariance}
    If $X$ and $Y$ act on disjoint spaces, then $\Cov_{\rho\otimes\sigma}[X,Y]=\Cov[X,Y]=0$.
\end{corollary}
To recap the proof of the BOW estimator (not in full generality), we assume one is given $m$ copies of an unknown quantum state $\rho$ and $m$ copies of a known quantum state $\sigma$. Then, the observable $\mathcal{O}_{(\rho\rho)}-2\mathcal{O}_{(\rho\sigma)}+\mathcal{O}_{(\sigma\sigma)}$ is an estimator for the HS distance between $\rho$ and $\sigma$. To explain this notation, recall this observable will be measured on the product state $\rho^{\otimes m}\otimes\sigma^{\otimes m}$, which we can identify with indices in $\{1,...,m,m+1,...,2m\}$. Each subscript identifies an orbit, which is a collection of cycles. For example $(\rho\rho)$ identifies the orbit which is all cycles which swap one index in $\{1,...,m\}$ with another, and $(\rho\sigma)$ identifies the orbit which is all cycles which swap one index in $\{1,...,m\}$ with one in $\{m+1,...2m\}$. Importantly, Badescu \etal~show this observable estimates the trace
\begin{lemma}[Corollary 4.24, \cite{DBLP:conf/stoc/BadescuO019}]
    $\mathcal{O}_{(\rho\sigma)}$ is an (efficient) estimator for $\tr[\rho\sigma]$.
\end{lemma}
This intuitively explains why the above operator estimates the squared Hilbert--Schmidt distance, since $\|\rho-\sigma\|_2^2=\tr[\rho^2]-2\tr[\rho\sigma]+\tr[\sigma^2]$. It is also worth pointing out that the subscripts denote how many copies of a state is used in one instance of the observable. That is the whole observable, for example $\mathcal{O}_{(\rho\rho)}$, operates on the product register above, but it is an average of all cycles in the orbit, and each cycle in this orbit acts on two copies of the state $\rho$. Similarly, each cycle in the orbit of $\mathcal{O}_{(\rho\sigma)}$ acts on one element of $\rho$ and one element of $\sigma$.

We aim to mirror this logic then, to produce our CQ estimator, with some simplifications, since we do not need the most general form. The first simplification we will make is that our operators will not operate on copies of $\sigma$, the known state. This is because the central node has full knowledge of $\sigma$, and hence can create this themselves. Hence, assume we are given $m$ independent copies of the unknown CQ state $\rho$. This means that the estimator will be an observable on $(Q\otimes C)^{\otimes m}$, measured on the product state $\rho^{\otimes m}$. Let the projector onto classical string $c$ be $\Pi_c=|c\rangle\langle c|$. Then to create this estimator, we essentially need to mimic the three observables used by Badescu \etal~to create the overall estimator observable. We also introduce some more definitions and results that we will need, to take advantage of the fact that all of our states are CQ states.
\begin{definition}[Partial Expectation]
    Assume we have two copies of a quantum state $\rho$, and we have an operator $X_{1,2}$ that acts on two copies of the state. In other words, it acts on $\rho_1\otimes\rho_2$, where all $\rho_i$ are identical (this is just for clarity). Then, define the partial trace over the second copy as follows:
    \begin{equation*}
        \mathbb{E}_{2}[X_{1,2}]=\tr_2[(I_1\otimes\rho_2)X_{1,2}],
    \end{equation*}
    where $\tr_2$ represents taking the trace over the second copy. Note we have, as one would expect, that $\mathbb{E}_1[\mathbb{E}_2[X_{1,2}]]=\mathbb{E}_{\rho_1\otimes\rho_2}[X_{1,2}]$. When clear from now on, we will also drop the subscripts on $\rho_i$ when the states are identical.
\end{definition}
Again, one could try to generalize this to operators that work over $M$ copies of a state. However, we will only ever need to take the partial expectation of operators that work only on two copies of a state, so we do not bother with this. One particular case we will be concerned about is when we have two different observables, which each operate on two copies of a state, but they overlap on one index. We prove a result about this now:
\begin{lemma}\label{lem:partial_operator_variance}
    Assume we have $m$ copies of a state $\rho$. Let $H_{1,2}$ and $H_{1,3}$ be two observables each operating on two copies of a state $\rho$, on copies $(1,2)$ and $(1,3)$ respectively. Denote their mean $\mu=\mathbb{E}_{\rho\otimes \rho}[H_{a,b}]$. Note since all copies are identical, this mean is the same for all $a\neq b$, and hence we do not put a subscript on it. Denote $\widetilde{H}_{a,b}=H_{a,b}-\mu I$. Let $G_1=\mathbb{E}_i[H_{1,i}]$, and $\widetilde{G}_1=G_1-\mu I$. Then, 
    \begin{equation*}
        \mathbb{E}_{\rho^{\otimes 3}}[\widetilde{H}_{1,2}\widetilde{H}_{1,3}]=\mathbb{E}_\rho[\widetilde{G}_1^2]  =\Var_\rho[G_1].  
    \end{equation*}
\end{lemma}
\begin{proof}
    First, note we choose the indices $(1,2)$ and $(1,3)$ without loss of generality, since all copies of $\rho$ are the same. We can use any indices, so long as exactly one overlap, so for ease, we choose these. The observables share copy 1, and copies 2 and 3 of the state $\rho$ are independent under the product state $\rho_1\otimes\rho_2\otimes\rho_3$ (where the subscripts have been put in to indicate the different copies). Then, notice that taking the expectation over $\widetilde{H}_{1,i}$ the latter index in both $(1,2)$ and $(1,3)$ give $\widetilde{G}_1$. Then, we can compute the expectation:
    \begin{align*}
        \mathbb{E}_{\rho^{\otimes 3}}[\widetilde{H}_{1,2}\widetilde{H}_{1,3}]&=\mathbb{E}_1[\mathbb{E}_2[\mathbb{E}_3[\widetilde{H}_{1,2}\widetilde{H}_{1,3}]]]\\
        &=\mathbb{E}_1[\mathbb{E}_2[\widetilde{H}_{1,2}]\mathbb{E}_3[\widetilde{H}_{1,3}]]\\
        &=\mathbb{E}_1[\widetilde{G}_1^2]\\
        &=\mathbb{E}_\rho[\widetilde{G}_1^2].
    \end{align*}
    To get the second part of the equality, notice that $\mathbb{E}_\rho[G_1]=\mu$. Then,
    \begin{align*}
        \mathbb{E}_\rho[\widetilde{G}_1^2]&=\mathbb{E}_\rho[(G_1-\mu I)^2]\\
        &=\mathbb{E}_\rho[G_1^2]-2\mu I\mathbb{E}_\rho[G_1]+\mu^2I\\
        &=\mathbb{E}_\rho[G_1^2]-\mathbb{E}_\rho[G_1]^2\\
        &=\Var_\rho[G_1].
    \end{align*}
\end{proof}
We aim first to create the observable that uses two copies of $\rho$. In other words, we want to estimate $\tr[\rho^2_{CQ}]$. However, we make an observation here first. Note that 
\begin{equation*}
    \tr[\rho^2]=\tr[\rho_{CQ}^2]=\sum_c\tr[\rho^2_c].
\end{equation*}
This is because the classical registers are orthonormal, so when taking the square, the cross terms with differing classical registers will vanish, and the classical registers that are identical will equal out to $1$. Hence, it is sufficient to estimate this latter sum. To estimate just one term in the sum, we know the swap operator is sufficient, since the swap operator is just a cycle of two copies of $\rho$. Let $F_{Q_i,Q_j}$ be the operator that swaps the quantum registers $Q_i$ and $Q_j$ respectively. This corresponds to the orbit $(\rho\rho)$. However, in the CQ setting, we only want to perform this swap if the classical registers are the same. Hence, the full operator, if we consider a fixed classical string $c$, will be defined as
\begin{equation*}
    F_{Q_i,Q_j}\otimes\Pi_c^{C_i}\otimes\Pi_c^{C_j},
\end{equation*}
where abusing notation $\Pi_c^{C_i}$ is the projector projecting onto string $c$ the $i$-th copy of $\rho$. Here, we have also abused notation to rearrange the operators. Strictly speaking, this operator should act on a spaced $(Q_i\otimes C_i)\otimes(Q_j\otimes C_j)$, but for ease of notation, we have just shifted the classical registers to the end. Note this works for a fixed string $c$, and fixed indices $i$ and $j$. To estimate $\tr[\rho_c^2]$, one needs to take the average over all such swaps. In other words, the observable
\begin{equation*}
    \frac{1}{\binom{m}{2}}\sum_{i<j}F_{Q_i,Q_j}\otimes\Pi_c^{C_i}\otimes\Pi_c^{C_j}
\end{equation*}
estimates $\tr[\rho_c^2]$. Finally, this is for a fixed string $c$, but we need to estimate the entire sum over all $c$. Hence, fixing some notation, let
\begin{equation*}
    S_{ij}=\sum_{c=1}^KF_{Q_i,Q_j}\otimes \Pi_c^{C_i}\otimes \Pi_C^{C_j}.
\end{equation*}
Then, the observable that estimates $\tr[\rho^2]$ is given by 
\begin{equation*}
    \frac{1}{{\binom{m}{2}}}\sum_{i<j}S_{ij}.
\end{equation*}
The next analogue we try to draw is that of the $(\sigma\sigma)$ orbit, which is an estimate of $\tr[\sigma^2]$. However, here we simplify some of Badescu's work. Namely, instead of giving a general operator, we will exploit the fact that we have a full description of our state $\sigma$, and hence we can directly compute $\tr[\sigma^2]$. Thus, we do not need an operator to arrive at this quantity. Expressed as an operator, this is $\tr[\sigma^2]I$, where $I$ is the identity over the entire space $\rho^{\otimes m}$. Hence, our (incomplete) operator to estimate HS distance currently looks like
\begin{equation*}
    \underbrace{\frac{1}{\binom{m}{2}}\sum_{i<j}F_{Q_i,Q_j}\otimes\Pi_c^{C_i}\otimes\Pi_c^{C_j}}_{\mathcal{O}_{(\rho\rho)}}+\underbrace{\tr[\sigma^2]I}_{\mathcal{O}_{(\sigma\sigma)}}.
\end{equation*}
The final observable we need is the equivalent of the $\mathcal{O}_{(\rho\sigma)}$ term, which estimates the quantity $\tr[\rho\sigma]$. Again, for the same reason above, due to the orthonormality of the classical registers, it is sufficient to estimate the quantity $\sum_c\tr[\sigma_c\rho_c]$. Let us fix a string $c$. Normally, in Badesscu \etal's framework, this would correspond to the cycle that swaps an index of $\rho$ with an index of $\sigma$. However, we have chosen to not to include copies of $\sigma$ for simplicity. However, with full knowledge of $\sigma$, we can create an observable that mimics that. For instance, summing over all strings $c$ and a fix index $i$ of $\rho$, we can define the observable
\begin{equation*}
    L_i=\sum_c\sigma_c\otimes\Pi_c^{C_i},
\end{equation*}
which has expectation $\tr[\sigma_c\rho_c]$, as desired. Then, if we take the average over all such indices $i$, this gives an estimator for $\tr[\sigma_C\rho_C]$:
\begin{equation*}
    \frac{1}{m}\sum_{i=1}^mL_i.
\end{equation*}
Hence, we now have the observable
\begin{equation*}
    \hat\Gamma=\frac{1}{\binom{m}{2}}\sum_{i<j}S_{ij}-\frac{2}{m}\sum_{i=1}^mL_i+\tr[\sigma^2]I
\end{equation*}
We claim this is an estimator for $\Gamma=D_{HS}^2(\rho,\sigma)$ with the variance as in \cref{thm:cq_hs_estimator}. First, we prove it has the required expectation.
\begin{lemma}\label{lemma:cq_hs_estimator_expectation}
    The above estimator satisfies $\mathbb{E}[\hat\Gamma]=\Gamma$.
\end{lemma}
\begin{proof}
    We can individually calculate the expectations of each operator. First, let us compute the expectation of $S_{1,2}$, which is equivalent to the expectation of every other pair. This gives us that
    \begin{equation*}
        \mathbb{E}[S_{1,2}]=\sum_c\tr[(\rho_c\otimes\rho_c)F_{Q_1,Q_2}].
    \end{equation*}
    Then, it is known for the swap operator that $\tr[(A\otimes B)F]=\tr[AB]$, and hence this gives that $\mathbb{E}[S_{1,2}]=\sum_c\tr[\rho_c^2]$, which is equal to $\tr[\rho^2]$. Next, we compute the expectation of $L_1$, which has the same expectation as all other $L_i$. This is simply given by
    \begin{equation*}
        \mathbb{E}[L_1]=\tr[L_1\rho]=\sum_c\tr[\rho_c\sigma_c].
    \end{equation*}
    Then, the expectation the last operator is a constant,and hence has expectation $\tr[\sigma^2]=\sum_c\tr[\sigma_c^2]$. Then, putting all three together, the expectation of the CQ estimator is
    \begin{align*}
        \mathbb{E}[\hat\Gamma]&=\frac{1}{\binom{m}{2}}\sum_{i<j}\sum_c\tr[\rho_c^2]+\frac{2}{m}\sum_c\tr[\rho_c\sigma_c]+\sum_c\tr[\sigma^2_c]\\
        &=\sum_c\tr[(\rho_c-\sigma_c)^2]\\
        &=\sum_c\tr[\Delta^2]\\
        &=\Gamma.
    \end{align*}
\end{proof}
We next prove that is the required variance. To do this, we will rewrite our operator notation. Let
\begin{equation*}
    H_{i,j}=S_{i,j}-L_i-L_j+\tr[\sigma^2]I.
\end{equation*}
It is then easy to check that $\hat\Gamma=\frac{1}{\binom{m}{2}}\sum_{i<j}H_{i,j}$. We now prove the following result on the variance.
\begin{lemma}\label{lemma:cq_hs_estimator_variance}
    Let $\mu=\mathbb{E}[H_{i,j}]$, and $\widetilde{H}_{i,j}=H_{i,j}-\mu I$. Let $G_i=\mathbb{E}_j[H_{i,j}]$, the expectation over just the second term in $H_{i,j}$. Define $A_1=\Var_\rho(G_i)$ and $A_2=\mathbb{E}_{\rho\otimes\rho}[H_{1,2}^2]$. Then,
    \begin{equation*}
        \Var(\hat\Gamma)\leq O\left(\frac{A_1}{m}+\frac{A_2}{m^2}\right).
    \end{equation*}
\end{lemma}
\begin{proof}
    Since shifts by the mean do not change variance, we have that
    \begin{equation*}
        \Var(\hat\Gamma)=\Var\left(\frac{1}{\binom{m}{2}}\sum_{i<j}H_{i,j}\right)=\Var\left(\frac{1}{\binom{m}{2}}\sum_{i<j}\widetilde{H}_{i,j}\right).
    \end{equation*}
    Since $\mathbb{E}[\widetilde{H}_{1,2}]=0$, this gives that
    \begin{equation*}
        \Var(\hat\Gamma)=\mathbb{E}\left[\left(\frac{1}{\binom{m}{2}}\sum_{i<j}\widetilde{H}_{i,j}\right)^2\right]=\frac{1}{\binom{m}{2}^2}\sum_{i<j}\sum_{k<l}\mathbb{E}[\widetilde{H}_{i,j}\widetilde{H}_{k,l}].
    \end{equation*}
    We now consider three cases, depending on the relationship between the pairs of tuples $(i,j)$ and $(k,l)$.

    \paragraph{Case 1: Disjoint tuples.} It the tuples are completely disjoint, then by \cref{cor:operator_covariance}, we have that $\mathbb{E}[\widetilde{H}_{i,j}\widetilde{H}_{k,l}]=\mathbb{E}[\widetilde{H}_{i,j}]\mathbb{E}[\widetilde{H}_{k,l}]=0$.

    \paragraph{Case 2: Identical tuples.} If $(i,j)=(k,l)$, then the contribution is given by $\mathbb{E}[\widetilde{H}_{i,j}^2]=\Var(H_{i,j})\leq \mathbb{E}[H_{i,j}^2]=A_2$. Since we have $\binom{m}{2}$ such terms, the total contribution to the variance is at most
    \begin{equation*}
        \frac{A_2}{\binom{m}{2}}\leq O\left(\frac{A_2}{m^2}\right).
    \end{equation*}

    \paragraph{Case 3: Pairs that overlap once.} Assume now the two pairs overlap on one index. By symmetry, it is sufficient to consider the contribution to the variance by pairs $H_{1,2}$ and $H_{1,3}$, and then sum over all pairs. Similarly, we can also equate $A_1=\Var(G_i)=\Var(G_1)$. Then, by \cref{lem:partial_operator_variance}, we have that $\mathbb{E}[\widetilde{H}_{1,2}\widetilde{H}_{1,3}]=A_1$. For each fixed tuple $(i,j)$, exactly $2(m-2)$ other tuples share exactly one index with it, with $m-2$ coinciding on index $i$, and the other $m-2$ coinciding on index $j$. Hence, by summing over all pairs, the total contribution to the variance in this case is at most
    \begin{align*}
        \frac{2(m-2)}{\binom{m}{2}}A_1&\leq O\left(\frac{A_1}{m}\right).
    \end{align*}
    Hence, summing all three cases, we have that the variance of the estimator is at most
    \begin{equation*}
        \Var(\hat\Gamma)\leq O\left(\frac{A_1}{m}+\frac{A_2}{m^2}\right).
    \end{equation*}
\end{proof}
The task now, is to bound the quantities $A_1$ and $A_2$. We begin by bounding $A_1$. For ease of notation, let 
\begin{equation*}
    L_\rho=\sum_c\rho_c\otimes \Pi_c^C.
\end{equation*} 
Then, we prove the following bound
\begin{lemma}\label{lemma:A_1_bound}
    $A_1\leq r_\infty\Gamma$.
\end{lemma}
\begin{proof}
    We can look at $H_{1,2}$ wlog. Recall that $G_1$ is the partial expectation over the second index, and hence we can compute the partial expectation over each individual term in $H_{1,2}$. We start with $\mathbb{E}_2[S_{1,2}]$:
    \begin{align*}
        \mathbb{E}_2[S_{1,2}]&=\tr_2[(1\otimes\rho)S_{1,2}]\\
        &=\tr_2\left[\sum_c(I\otimes\rho_c)F_{Q_1,Q_2}\otimes|c\rangle\langle c|_{C_1}\otimes|c\rangle\langle c|_{C_2}\right]\\
        &=\tr_{Q_2}\left[\sum_c(\rho_C\otimes I)\otimes|c\rangle\langle c|_{C_1}\right]\\
        &=\sum_c\rho_C\otimes|c\rangle\langle c|_{C_1}=L_\rho.
    \end{align*}
    Here, we have abused notation to let $C_i$ and $Q_i$ denote the classical and quantum registers respectively of index $i$. Second, it is clear that $\mathbb{E}_2[L_1]=L_1$, since it is independent of the second register. We also get this equals $L_\sigma$, since $L_1$ is just $L_\sigma$ acting only on the first index. On the other hand, we have that
    \begin{align*}
        \mathbb{E}_2[L_2]&=\tr_2[(I\otimes\rho)(I\otimes L_\sigma)]\\
        &=\tr_2\left[I\otimes(\rho L_\sigma)\right]\\
        &=I_1\tr[L_\sigma\rho],
    \end{align*}
    which is a scalar multiple of the identity. We get this because $L_2$ is simply the observable $L_\sigma$ in index $2$. Finally, we note $\tr[\sigma^2]I$ also becomes a scalar multiple of the identity after tracing out the second copy. Hence, we have that $G_1=L_\rho-L_\sigma-\alpha I$, for some constant $\alpha$, which we can ignore when computing the variance, since scalar shifts do not impact the variance. Then, we have that
    \begin{align*}
        L_\rho-L_\sigma&=\sum_c(\rho_c-\sigma_c)\otimes\Pi_c=\sum_c\Delta_c\otimes\Pi_c\\
        (L_\rho-L_\sigma)^2&=\sum_c\Delta_c^2\otimes\Pi_c\\
        \Var_\rho(L\rho-L_\sigma)&\leq \mathbb{E}_\rho[(L_\rho-L_\sigma)^2]\\
        &=\sum_c\tr(\rho_c\Delta_c^2)\\
        &\leq \sum_c r_c\tr[\Delta_c^2]\\
        &\leq r_\infty\Gamma,
    \end{align*}
    where we use that by definition $0\preceq\rho_c\preceq p_cI$. This then gives that $A\leq p_\infty\Gamma$.
\end{proof}
We now bound $A_2$:
\begin{lemma}\label{lemma:A_2_bound}
    $A_2\leq \chi_r+\chi_s$.
\end{lemma}
\begin{proof}
    We start by an operator inequality. By the operator Jensen's inequality \cite[Theorem 2.1]{HANSEN_2003}, for Hermitian operators $B_i$, we have that
\begin{equation*}
    (B_1+B_2+B_3+B_4)^2\preceq O\left(B_1^2+B_2^2+B_3^2+B_4^2\right),
\end{equation*}
with slight abuse of big-Oh notation (that is, there is some constant C such that this inequality holds). Recalling that $H_{1,2}$ is a sum of four operators, it suffices to bound the second moment of each of the operators. First, consider $S_{1,2}^2$, which simply applies the swap twice. Hence, 
\begin{equation*}
    S_{1,2}^2=\sum_cI_{Q_1,Q_2}\otimes \Pi_c^{C_1}\otimes \Pi_c^{C_2}.
\end{equation*}
Hence, we have that $\mathbb{E}[S_{1,2}^2]=\sum_c r_c^2=\chi_r$. Next, consider $L_\sigma^2=\sum_c\sigma_c^2\otimes\Pi_c$.Note that since $\sigma_c$ is positive semidefinite, and $\tr[\sigma_c]=s_c$, every eigenvalue of $\sigma_c$ is at most $s_c$. Hence, we have that $\sigma_c^2\preceq s_c^2I$. Then, $\tr(\rho_c\tau_c^2)\leq r_cs_c^2\leq s_c^2$. Summing over all possible $c$, this gives $\mathbb{E}[L_\sigma^2]\leq\sum_c s_c^2=\chi_s$. This bound applies to both $L_1^2$ and $L_2^2$.Then, by definition, we have that $\tr[\sigma^2]I\leq \sum_cs_c^2=\chi_s$. Since $\chi_s\leq 1$, the second moment of this operator is also less than $\chi_s$. Combining these four estimates along with the operator Jensen's inequality gives $A_2\leq O(\chi_r+\chi_s)$.
\end{proof}
Combining the variance bound, along with the bounds on $A_1$ and $A_2$, we prove \cref{thm:cq_hs_estimator}.
\section{Moment Bounds}\label{app:moment_bounds}
The aim of this section is to prove \cref{lemma:first_moment_upsilon} and \cref{lemma:upsilon_second_moment}. We start by proving some small results that we will need to prove the former.
\subsection{First Moment Bound}
\begin{lemma}\label{lemma:swap_branch_representation}
    Recalling the definition of $A_c(U)$ from \cref{def:Haar_instrument}, we have that $\|\mathcal{I}_{U,c}(X)\|_2^2=\tr[(A_c(U)\otimes A_c(U))(F_{Q_1,Q_2}\otimes I_{B_1,B_2})]$ for every classical label $c$.
\end{lemma}
\begin{proof}
    Using that $\tr[(A\otimes B)F]=\tr[AB]$, we have that 
    \begin{equation*}
        \|\tr_B[A_c(U)]\|_2^2=\tr[(\tr_B[A_c(U)\otimes\tr_B(A_c(U))F_{Q_1,Q_2}])]=\tr[(A_c(U)\otimes A_c(U))(F_{Q_1,Q_2}\otimes I_{B_1,B_2})].
    \end{equation*}
    The second equality we get by recalling that we can pull back partial traces, in the sense that $\tr[(\tr_{A,B}Z)Y)=\tr[Z(Y\otimes I_{B_1,B_2})]$, where $I_{B_1,B_2}$ is shorthand for $I_{B_1}\otimes I_{B_2}$.
\end{proof}
Using this, we can get another way to express $\Upsilon_U(X)$,
\begin{lemma}\label{lemma:upsilon_expression}
    Define an operator acting on $(C_1Q_1B_1)\otimes(C_2Q_2B_2)$ by
    \begin{equation*}
        R=\sum_{c=1}^K|c\rangle\langle c|_{C_1}\otimes |c\rangle\langle c|_{C_2}\otimes F_{Q_1,Q_2}\otimes I_{B_1,B_2}.
    \end{equation*}
    Then, for every traceless Hermitian $X$, $\Upsilon_U(X)=\tr[R(UXU^\dagger)^{\otimes 2}]$.
\end{lemma}
\begin{proof}
    Recall \cref{lemma:swap_branch_representation} gives the contribution for each branch of a classical label $c$, over the space of registers $QB$. We can embed this back into the full $CQB$ space by inserting the projectors corresponding to the classical registers, for each copy. Then, we can take the sum over all $c$, which gives that 
    \begin{equation*}
        \Upsilon_U(X)=\tr[R(UXU^\dagger)^{\otimes 2}]
    \end{equation*}
\end{proof}
This second moment expectation is already something we handled previously in \cref{lemma:Haar_pair}, so we follow the analysis there. \begin{lemma}\label{lemma:traceless_second_moment}
    Let $X$ be a traceless Hermitian, and let $s=\|X\|_2^2$. Then, 
    \begin{equation*}
        \mathbb{E}_U[(UXU^\dagger)^{\otimes 2}]=-\frac{s}{d(d^2-1)}I+\frac{s}{d^2-1}F
    \end{equation*}
\end{lemma}
\begin{proof}
    By the same argument as in \cref{lemma:Haar_pair}, we get that the expectation equals $aI+bF$, for some constants. We again try to equate traces. Since $X$ is traceless, we have that $\tr[X\otimes X]=0$. which gives that
    \begin{equation*}
        ad^2+bd=0.
    \end{equation*}
    We can then use the swap operator again to get that
    \begin{equation*}
        ad+bd^2=s.
    \end{equation*}
    Solving both of these gives the desired constants.
\end{proof}
We can also find the trace of the operator $R$:
\begin{lemma}\label{lemma:R_trace}
    We have that $\tr[R]=Kd_qb^2$ and $\tr[RF]=Kd_q^2b$. Here $F$ denotes the full swap between $C_1Q_1B_1$ and $C_2Q_2B_2$.
\end{lemma}
\begin{proof}
    For a fixed value of $c$, the term inside the sum of $R$ is given by $|c\rangle\langle c|_{C_1}\otimes |c\rangle\langle c|_{C_2}\otimes F_{Q_1,Q_2}\otimes I_{B_1,B_2}$. This has trace equal to $\tr(F_{Q_1,Q_2})\cdot\tr[I_{B_1,B_2}]=d_qb^2$. Summing over all $K$ labels gives $\tr[R]=Kd_qb^2$. For the second, we can view the overall swap operator as a tensor of swaps $F=F_C\otimes F_Q\otimes F_B$. For a fixed label $c$, we have that $\tr[c\rangle\langle c|\otimes |c\rangle\langle c|F_C]=1$. Next, $\tr[F_{Q_1,Q_2}F_Q]=\tr[I_{Q_1,Q_2}]=d_q^2$. Finally, $\tr[I_{B_1,B_2}F_B]=\tr[F_B]=b$. Then, summing over all $K$ labels gives $\tr[RF]=Kd_q^2b$.
\end{proof}
We can now give the proof of \cref{lemma:first_moment_upsilon}.
\begin{proof}
    Starting from \cref{lemma:upsilon_expression}, we get that $\Upsilon_U(X)=\tr[R(UXU^\dagger)^{\otimes 2}]$. Then, we can take expectations, and apply \cref{lemma:traceless_second_moment}, to get that
    \begin{equation*}
        \mathbb{E}_U[\Upsilon_U(X)]=\tr\left[R\left(-\frac{s}{d(d^2-1)}I+\frac{s}{d^2-1}F\right)\right].
    \end{equation*}
    We can bound the operators with \cref{lemma:R_trace} to get that this quantity equals
    \begin{align*}
        \mathbb{E}_U[\Upsilon_U(X)]&=-\frac{s}{d(d^2-1)}Kd_qb^2+\frac{s}{d^2-1}Kd_q^2b\\
        &=\frac{s}{(d^2-1)}b+\frac{s}{d^2-1}dd_q\\
        &=\frac{d_qd-b}{d^2-1}\|X\|_2^2.
    \end{align*}
    To see the final inequality, note $b=\frac{d}{Kd_q}$. Then, we can rewrite $d_qd-b=d(d_q-1/(Kd_q))$. Assuming $Kd_q^2\geq 2$, we have that $q-1/(Kd_q)\geq d_q/2$, and hence $dd_q-b\geq dd_q/2$. Then, using that $d^2-1\leq d^2$, we get the stated inequality.
\end{proof}
\subsection{Permutation Facts}
Before proving the second moment bound of our operator $\Upsilon_U$, we will need to analyze permutations and Weingarten coefficients, like Doosti \etal. Hence, we lay out some results we will use. The first result we use is the same Weingarten expression for the expectation, also courtesy of Mele \cite{DBLP:journals/quantum/Mele24}:
\begin{lemma}\label{lemma:Weingarten_expression}
    Given a permutation $\pi\in S_k$, let $P_\pi\in \mathbb{C}^{d_k\times d_k}$ be the permutation operator and $Wg(\pi,d)$ be the Weingarten coefficient. Then, for a state M,
    \begin{equation*}
        \mathbb{E}_U[U^{\otimes k}MU^{\dagger\otimes k}]=\sum_{\pi,\kappa\in S_k}Wg(\pi^{-1}\kappa,d)\tr[P_\kappa^\dagger M]P_\pi.
    \end{equation*}
\end{lemma}
Now, we assume our permutations are in the symmetric group $S_4$, and introduce some notation.
\begin{definition}[Permutation notation]
    For a permutation $\pi$, we let $|\cyc(\pi)|$ be the number of cycles in the permutation (including singleton cycles). Further, let $|\pi|=4-|\cyc(\pi)|$ be the \emph{transposition length}.
\end{definition}
For example, $\pi=(12)(3)(4)$ has $|\cyc(\pi)|=3$ and $|\pi|=1$. Given four copies of a state $\rho$, with each copy defined on the space $H=C\otimes Q\otimes B$, and for a given permutation $\pi\in S_4$, we can define a permutation operator that swaps the copies according to the permutation $\pi$:
\begin{equation*}
    P_\pi^H=P_\pi^C\otimes P_\pi^Q\otimes P^B_\pi,
\end{equation*}
where the operators on the right swap over specific registers of the overall quantum states $\rho$. For instance, if $\pi=(12)(3)(4)$, the the operator $P_\pi$ swaps the first two copies of $\rho$, and keeps the third and fourth copies untouched (that is, it is the standard swap operator). We specifically let $\alpha=(12)(34)$, the double swap permutation. Note when expanded to four copies of a state, we can denote $R$ from \cref{lemma:upsilon_expression} as $R_{1,2}$ and $R_{3,4}$ to operate on the pairs. Then, we let $\mathbf{R}=R_{1,2}\otimes R_{3,4}$. Then, we have that our second moment is given by
\begin{equation*}
    \Upsilon_U(X)^2=\tr[\mathbf{R}(UXU^\dagger)^{\otimes 4}].
\end{equation*}
We then give the following result on the trace of a permutation.
\begin{lemma}\label{lemma:permutation_trace}
    Let $\pi\in S_4$ be a permutation, and $\gamma$ a cycle in $\cyc(\pi)$. Let $\ell(\gamma)$ denote the length of the cycle. Let $A$ be an operator on a finite dimensional Hilbert space. Then, we have
    \begin{equation*}
        \tr[P_\pi A^{\otimes 4}]=\prod_{\gamma\in\cyc(\pi)}\tr[A^{\ell(\gamma)}].
    \end{equation*}
\end{lemma}
One can think of this as a special case of the swap property on traces, and it itself is a special case of a result of Collins and Nechita, whom define it for $n$ non-identical operators \cite[Section 2]{Collins_2010}. We also make use of another result by Collins and Nechita, that bounds the Weingarten function \cite[Section 2]{Collins_2010}:
\begin{lemma}\label{lemma:Weingarten_bound}
    For every permutation $\pi\in S_4$, there is an absolute constant $C_W$ such that $|W_g(\pi;d)|\leq C_wd^{-4-|\pi|}$. In other words, $|W_g(\pi;d)|\leq O(d^{-4-|\pi|})$.
\end{lemma}
We can prove one more small result about the composition of permutations:
\begin{lemma}\label{lemma:cycle_permutation}
    Let $\pi_1\in S_4$ be a permutation with no fixed points, and $\pi_2\in S_4$ be a permutation. Then, letting, $\pi_3=\pi_1^{-1}\pi_2$, we have that
    \begin{equation*}
        |\cyc(\pi_2|)|\leq 2+|\pi_3|.
    \end{equation*}
\end{lemma}
\begin{proof}
    Note since $|\pi|$ satisfies the triangle inequality, and hence we have that
    \begin{equation*}
        |\pi_2|=|\pi_1\pi_3|\geq |\pi_1|-|\pi_3|.
    \end{equation*}
    Since $\pi_1$ has no fixed points it either is two swap cycles, or a cycle between all four elements. Hence, $|\cyc(\pi_1)\leq 2$. This implies $|\pi_1|\geq 2$. We get the result by rearranging.
\end{proof}
We now need a specific result about the action of the permutation operator with our operator $\mathbf{R}$.
\begin{lemma}\label{lemma:R_permutation_bound}
    For every $\pi\in S_4$, we have that $|\tr[\mathbf{R}P_\pi]|\leq K^2d_q^4b^{|\cyc(\pi)|}$.
\end{lemma}
\begin{proof}
    To write out the operators explicitly, let us denote
    \begin{equation*}
        \mathbf{R}_C=\sum_{c,d}|c,c,d,d\rangle\langle c,c,d,d,|_{C_1,C_2,C_3,C_4},
    \end{equation*}
    which is a projector of rank $K^2$. Then, recall $\mathbf{R}=R_{1,2}\otimes R_{3,4}$ which factors across the registers as $\mathbf{R}=\mathbf{R}_C\otimes P^Q_\alpha\otimes I_{B_1,B_2,B_3,B_4}$. Then, we can factor this among the traces of $P_\pi$ to get
    \begin{equation*}
        \tr[\mathbf{R}P_\pi]=\tr[\mathbf{R}_CP_\pi^C]\tr[P_\alpha^QP_\pi^Q]\tr[P_\pi^B].
    \end{equation*}
    We now aim to bound each individual term. Since each permutation operator is also a unitary, we have that
    \begin{equation*}
        |\tr[\mathbf{R}_CP_\pi^C]|\leq |\tr[\mathbf{R}_C]|=K^2.
    \end{equation*}
    For the second term, we apply the cycle trace identity to get that
    \begin{equation*}
        \tr[P_\alpha^QP_\pi^Q]=\tr[P_{\alpha\pi}^Q]=d_q^{|\cyc(\alpha\pi)|}\leq d_q^4,
    \end{equation*}
    where implicitly, the permutations are working on the operator $I_Q^{\otimes 4}$. With the same argument, we get that
    \begin{equation*}
        \tr[P_\pi^B]=b^{|\cyc(\pi)|}.
    \end{equation*}
    Putting this all together, we get that $|\tr[\mathbf{R}P_\pi]|\leq K^2d_q^4b^{|\cyc(\pi)|}$.
\end{proof}
With all this setup out of the way, we can finally prove the second moment bound, \cref{lemma:upsilon_second_moment}.
\begin{proof}
    Recall by definition, $\Upsilon_U(X)^2=\tr[\mathbf{R}(UXU^\dagger)^{\otimes 4}]$, and thus taking expectations, we have that
    \begin{equation*}
        \mathbb{E}_U[\Upsilon_U(X)^2]=\tr[\mathbf{R}\mathbb{E}_U(UXU^\dagger)^{\otimes 4}].
    \end{equation*}
    Using \cref{lemma:Weingarten_expression}, we get an expression for this expectation
    \begin{equation*}
        \mathbb{E}_U[(UXU^\dagger)^{\otimes 4}]=\sum_{\pi,\kappa\in S_k}Wg(\pi^{-1}\kappa,d)\tr[P_\pi^\dagger X^{\otimes 4}]P_\kappa.
    \end{equation*}
    Then, applying the operator $\mathbf{R}$ and taking traces, we can recover $\Upsilon^2_U(X)$. Note on the right hand side of the above expression, the only nonconstant term is the $P_\pi$ operator, so the trace and operator can be applied to just this term:
    \begin{equation*}
        \mathbb{E}_U[\Upsilon_U(X)^2]=\sum_{\pi,\kappa\in S_k}Wg(\pi^{-1}\kappa,d)\tr[P_\pi^\dagger X^{\otimes 4}]\tr[\mathbf{R}P_\kappa].
    \end{equation*}
    We can use \cref{lemma:permutation_trace} to bound the middle term first. This gives that
    \begin{equation*}
        \tr[P_\pi^\dagger X^{\otimes 4}]=\prod_{\gamma\in\cyc(\pi)}\tr[X^{l(\gamma)}].
    \end{equation*}
    Since $X$ is traceless, this will only be nonzero for permutations with no singletons, which is either two cycles of length two, or a cycle of length four. In the first case, we get $\tr[X^2]^2=\|x\|_2^4$. In the latter, we get $\tr[X^4]$, which is at most $\|X\|_2^4$ for Hermitian $X$. Hence in either case, we get
    \begin{equation*}
        \tr[P_\pi^\dagger X^{\otimes 4}]\leq \|X\|_2^4.
    \end{equation*}
    Using \cref{lemma:permutation_trace}, we can also bound the latter term, which is at most $K^2d_q^4b^{|\cyc(\kappa)|}$. For some $\pi$ that contributes, and for some arbitrary $\kappa$, we can set $\upsilon=\pi^{-1}\kappa$. Then, by \cref{lemma:cycle_permutation}, we get that $|\tr[\mathbf{R}P_\kappa]|\leq K^2d_q^4b^{2+|\upsilon|}\leq K^2d_q^{4+|\upsilon|}b^{2+|\upsilon|}$, since $d_q\geq 1$. We can use the Weingarten bound in \cref{lemma:Weingarten_bound} to get that $|Wg(\upsilon;d)|\leq C_Wd^{-4-|\upsilon|}$. Hence using this and that $d=Kd_qb$, we get that
    \begin{align*}
        |Wg(\upsilon;d)||\tr[\mathbf{R}P_\kappa]|&\leq C_W\frac{K^2d_q^{4+|\upsilon|}b^{2+|\upsilon|}}{d^{4+|\upsilon|}}\\
        &=C_W\frac{1}{K^{2+|\upsilon|}b^2}\\
        &\leq C_W\frac{1}{K^2b^2}=C_W\frac{d_q^2}{d^2}.
    \end{align*} We can then sum over all possible choices of $\pi$ and $\kappa$. Only $9$ such choices of $\pi$ are nonzero, and there are $24$ choices of $\kappa$, so summing over all of these, which is a constant amount, we get that
    \begin{equation*}
        \mathbb{E}_U[\Upsilon_U(X)^2]\leq O\left(\frac{d_q^2}{d^2}\|X\|_2^4\right).
    \end{equation*}
    The explicit constant can be taken as $\tilde{C}=9\cdot24\cdot C_W$.
\end{proof}
\end{document}